\documentclass{article}

\IfFileExists{MinionPro.sty}
{\usepackage[lf]{MinionPro}
\usepackage{amsmath,amsthm,amsbsy}}
{\usepackage{times}
\usepackage{savesym} 
\usepackage{amsmath,amsthm,amsbsy}
\savesymbol{iint}
\savesymbol{openbox}
\usepackage{txfonts}
\restoresymbol{TXF}{iint}
\restoresymbol{TXF}{openbox}}

\usepackage{mathtools}
\usepackage{geometry}
\usepackage{color}
\definecolor{labelkey}{rgb}{0,0,.75}
\definecolor{MyGreen}{rgb}{0,.6,.2}
\definecolor{MyDarkBlue}{rgb}{.1,.1,.75}
\usepackage[colorlinks=true, citecolor=MyGreen, linkcolor=MyDarkBlue]{hyperref}
\usepackage{tensor}
\usepackage[all,cmtip]{xy}
\usepackage{graphicx}

\DeclareMathOperator{\ck}{\bf L}
\renewcommand{\div}{\mathop{\rm div}\nolimits}

\DeclareMathOperator{\Lap}{\Delta}

\DeclareMathOperator{\tr}{\rm tr}

\DeclareMathOperator{\extd}{\mathbf d}

\DeclareMathOperator{\Ker}{\mathrm{Ker}}
\DeclareMathOperator{\grad}{\mathrm{grad}}
\DeclareMathOperator{\Drift}{\mathrm{Drift}}

\def\ip<#1,#2>{\left<#1,#2\right>}
\let\ol\overline

\newcommand{\ra}{\rightarrow}

\newcommand{\Reals}{\mathbb{R}}
\newcommand{\Nats}{\mathbb{N}}

\newcommand{\calC}{\mathcal{C}}

\newcommand{\calE}{\mathcal{E}}
\newcommand{\calF}{\mathcal{F}}

\newcommand{\calJ}{\mathcal{J}}
\newcommand{\calK}{\mathcal{K}}

\newcommand{\calO}{\mathcal{O}}
\newcommand{\calP}{\mathcal{P}}
\newcommand{\calQ}{\mathcal{Q}}

\newcommand{\calU}{\mathcal{U}}

\newcommand{\olg}{\overline{g}}

\newcommand{\Stt}{S_{\mathrm{tt}}}
\newcommand{\RR}{\mathbb R}

\newcommand{\bfsigma}{\boldsymbol{\sigma}}
\def\dc<#1,#2,#3>{\{#1;\;#2,#3\}}

\usepackage{etoolbox}
\newcommand\drift[2][]{%
  \ifstrempty{#1}{%
    [#2]^\mathrm{drift}
  }{%
    [#2]^\mathrm{drift}_{#1}
  }%
}

  \newcounter{mnote}
  \let\oldmarginpar\marginpar
  \renewcommand\marginpar[1]{\-\oldmarginpar[\raggedleft\scriptsize #1]%
  {\raggedright\scriptsize #1}}

\title{Conformal Fields and the Structure of the Space of Solutions of the Einstein Constraint Equations}

\date{}

\author{Michael Holst \\ UC San Diego \and David Maxwell \\ University of Alaska \and Rafe Mazzeo \\ Stanford University}

\begin{document}
\newtheorem{theorem}{Theorem}[section]
\newtheorem{conjecture}[theorem]{Conjecture}
\newtheorem{problem}[theorem]{Problem}
\newtheorem{proposition}[theorem]{Proposition}
\newtheorem{corollary}[theorem]{Corollary}
\newtheorem{remark}[theorem]{Remark}
\newtheorem{lemma}[theorem]{Lemma}
\theoremstyle{definition}
\newtheorem{definition}[theorem]{Definition}
\numberwithin{equation}{section}

\maketitle

\begin{abstract}
The drift method, introduced in \cite{Maxwell:2014Drift}, provides a new formulation of the Einstein constraint equations, 
either in vacuum or with matter fields. The natural of the geometry underlying this method compensates for its slightly 
greater analytic complexity over, say, the conformal or conformal thin sandwich methods. We review this theory here
and apply it to the study of solutions of the constraint equations with non-constant mean curvature. We show that this
method reproduces previously known existence results obtained by other methods, and does better in one important
regard. Namely, it can be applied even when the underlying metric admits conformal Killing (but not true Killing) vector fields.
We also prove that the absence of true Killing fields holds generically. 
\end{abstract}

\section{Introduction}\label{sec:intro}
Let $(M, g, K)$ denote a triplet consisting of an $n$-dimensional manifold $M$, a metric $g$ on $M$, and 
an auxiliary symmetric $2$-tensor $K$. The vacuum Einstein constraint equations for this triplet are
\begin{subequations}\label{eq:constraints}
\begin{alignat}{2}
R_{ g} - | K|_{ g}^2 + (\tr_{ g}  K)^2 &= 0 &\qquad&\text{\small[Hamiltonian constraint]}\label{hamvac}\\
\div_{ g}  K- \extd (\tr_{g} K) &= 0. &\qquad&\text{\small[momentum constraint]}\label{momvac}
\end{alignat}
\end{subequations}
We typically assume that $M$ is compact, or at least that $(M,g)$ is complete. Solutions correspond to space-like hypersurfaces 
in a Lorentzian spacetime $(X,G)$, i.e., solutions of the vacuum Einstein equations $\mathrm{Ric}(G) = 0$, so $g$ is the 
induced metric and $K$ the second fundamental form of this hypersurface. Solutions to system \eqref{eq:constraints} serve 
as Cauchy data for the Einstein evolution problem (which of course must be supplemented by some choice of gauge to make
the problem hyperbolic). The interest in finding solutions of the constraint equations is directly tied in this way to the 
study of the general Einstein equations. More general versions of these equations include a cosmological constant and source 
terms, and will be recalled below.

The set of pairs $(g,K)$ which solve (\ref{hamvac}, \ref{momvac}) is infinite dimensional, and in a suitable
topology constitutes a Banach manifold (at least away from the solutions for which the linearized operator has cokernel).
To turn the search for these solutions into a less underdetermined and
hence more tractable problem, it is customary to decompose the space of all pairs $(g,K)$ into `slices' and consider the 
constraint equations as an equation 
within each slice. If done correctly, this leads to a family of semilinear elliptic equations, one for each slice, 
to which one can apply a vast panoply of known techniques.  The traditional slicing is known as the
conformal method, originally proposed by Lichnerowicz and Choquet-Bruhat, and studied by them and many others over the 
past 60 years.  
Another common method appearing in the intervening years is called the conformal thin sandwich method. Although apparently different, it
was proved by the second author \cite{Maxwell:2014gx} that this is completely equivalent to the older conformal method.

In the conformal method, the data for the slices consist
of triplets $(g, \tau, \sigma; N)$ where $g$ dictates the conformal class $[g]$ of the solution metric $\overline g$, $\tau$ is the `mean curvature function', i.e., $\tau=\tr_{\overline g} K$ for
the eventual solution, $\sigma$ is a transverse-traceless 
(i.e. trace-free and divergence-free) tensor with respect to $g$,
and $N$ is a positive function that plays the role of a gauge choice
and is related to the so-called lapse associated with a coordinate system
on the spacetime generated by the solution of the constraint equations.
A comprehensive description of solutions to the conformal method is known in the special case when $\tau$ is constant \cite{Isenberg:1995bi},
and this led to perturbative results shortly thereafter 
\cite{IsenbergMoncrief1996}. Significant breakthroughs were obtained by the first and
later the second authors \cite{HNT07a,M09} concerning existence for `far-from-CMC' data, where the mean-curvature function is
allowed to be variable and seemingly nowhere close to constant, with a price of requiring the transverse-traceless tensor to be
very small. This led to several new developments, and extensions and refinements of these ideas in various other standard settings.
It was pointed out recently, however, by Gicquaud and his collaborators \cite{Gicquaud:2014bu} that upon recasting the setup in certain
way, all of these results are still fundamentally perturbative and hence should be regarded as `near-CMC'. 

In recent years limitations of the conformal method in the far-from-CMC setting have appeared. We point to 
\cite{Maxwell:2011if} \cite{Maxwell:2014Kasner} along with the very nice results in \cite{Nguyen:2015}
(based on the original blowup analysis of \cite{Dahl:2012dk}) for examples where there exist either no or 
multiple solutions of the constraint equations corresponding to a given set of conformal data $(g,\sigma,\tau;N)$, and
there is nothing apparent in the geometry of this data set which allows one to a priori predict what happens.  Motivated by
these difficulties, the second author here proposed \cite{Maxwell:2014Drift} a different idea to slice up the space of pairs 
$(g,K)$. This is known as the {\it drift} method, and is based on an invariant geometric interpretation of the dynamics of 
spacelike hypersurfaces evolving in a Lorentzian Einstein manifold. We review these methods carefully
below. For now let us note one key difference. In the drift method, the mean curvature function $\tau$ is replaced by a 
pair $(\tau_*, V)$, where $\tau_*$ is a certain average $V$ is a vector field which represents a `drift' equivalence class.
The equations in this formulation are more nonlinear and more complicated than for the older methods, but the key motivation 
is that this new framework should make it easier to handle various well-known obstructions and subtleties in the conformal method.
More specifically, it is not clear how to make the conformal method work when the conformal class $[g]$ admits conformal Killing 
fields, and indeed, we show here that there is a fundamental breakdown in that procedure. That method is also less
tractable when $\tau$ has zeros.  In fact, there are no general a priori estimates for solutions of these equations, and there are 
examples of families of solutions which blow up. The hope remains that better methods may predict the data sets near which a priori 
estimates fail. 

The goal of the present paper is to show that drift method does at least as well as the conformal method, and in a certain sense, much
better. More specifically, we prove a set of existence results for the drift formulation of the constraint equations, both without and
with source terms, which include the far-from-CMC results cited above. All of this is done perturbatively around the CMC case.
The major improvement is that these results also hold when the conformal class $[g]$ admits conformal Killing fields, so long as
the the metric we are perturbing from has no Killing fields. 

This paper is organized as follows. We begin by reviewing the standard conformal method and introducing the notion of
conformal momentum, and then describe the precise way by which conformal Killing fields present an obstruction 
in the conformal method. We finally present the drift method in \S 4, and in \S 5 the adaptations necessary to incorporate 
matter fields.  Section 6 then proves the existence of near-CMC solutions using the drift method, and also establishes
that the hypotheses needed to apply this theorem hold generically.


\subsection{Notation and Conventions}\label{sec:notation}

In this paper we assume that $M$ is a manifold of dimension $n\ge 3$. 
We assume $M$ is compact, and occasionally do not say this explicitly in the statements of results, etc. 
Solutions to various equations are found in 
a Sobolev space $W^{k,p}$, where $k \in \Nats$, $k \geq 2$, and $p>1$ are chosen so that
\[
\frac{1}{p}-\frac{k-1}{n} < 0;
\]
this ensures that $W^{k,p}$ functions have H\"older continuous first derivatives. 
If $E$ is any smooth vector bundle over $M$, we write $W^{k,p}(M,E)$ for the space of sections of $E$
which are in $W^{k,p}$ with respect to any local trivialization. In particular, we have the bundles $TM$ of vector fields, 
$T^*M$ of covector fields, $S_2M$ of symmetric $(0,2)$ tensors and its subbundle $S_{\rm tt}(g)$ 
of transverse-traceless tensors with respect to the metric $g$.  Function spaces of positive
functions are denoted by a subscript $+$, e.g. $W^{k,p}_+(M)$.

We henceforth set the constants
\[
q = \frac{2n}{n-2}, \qquad \kappa = \frac{n-1}{n}, \qquad a = 2\kappa q,
\]
so $q$ is a critical Sobolev exponent and $\kappa$ and $a$ are dimensional constants which appear
in various equations below. 

We also consider the conformal Killing operator, whose action on vector fields is
\[
(\ck X)_{ab} = \nabla_a X_b + \nabla_b X_a - \frac{2}{n} \div X\, g_{ab}.
\]
Its adjoint $\ck^*$ acts on symmetric, trace-free $(0,2)$ tensor $A_{ab}$ by
\[
(\ck^* A)_b = -2\nabla^a A_{ab}
\]

The kernel of $\ck$ is the finite dimensional space $\mathcal Q$ of conformal Killing fields.


\section{The Standard Conformal Method and Conformal Momentum}\label{sec:confmeth}

The conformal method appears in the literature in two forms. 
The original conformal method was introduced by Lichnerowicz \cite{Lichnerowicz:1944}
and substantially extended by York, O'Murchadha, and Choquet-Bruhat among others in the 1970s.  
Some decades later York introduced the conformal thin-sandwich method \cite{YorkJr:1999jo}, and later,
with Pfeiffer, also gave an equivalent Hamiltonian formulation \cite{Pfeiffer:2003ka}.
It turns out that the original conformal method and the conformal thin-sandwich method
are really the same parameterization of the constraint equations \cite{Maxwell:2014gx}; we describe them 
here in a unified fashion that will also be helpful for describing the drift formulations of the constraint 
equations. For simplicity, we focus for now on the vacuum constraint equations; Section \ref{sec:matter}
below describes an approach for incorporating matter fields into both the standard conformal method and the drift formulation.

A metric and second fundamental form $(g,K)$ canonically determine 
\begin{itemize}
\item a conformal class $[g]$, and
\item a mean curvature $\tau = g^{ab}K_{ab}$.
\end{itemize}
These are two of the parameters of the conformal method. The third and final parameter is not completely
canonical and depends on a choice of volume form $\alpha$.  We will call $\alpha$ a \textit{volume gauge}. 
Once this has been fixed, the final parameter is 
\begin{itemize}
\item the conformal momentum of $(g,K)$ measured by $\alpha$,
\end{itemize}
which we define in Definition \ref{def:confmommeas} below.  We refer to \cite{Maxwell:2014gx} and 
\cite{Maxwell:2014Drift} for the geometric and physical motivation behind this terminology.
\begin{definition}[Conformal Momentum]\label{def:confmom}
A \textit{conformal momentum} is an equivalence class of pairs $(g,\sigma)$ where 
$g$ is a metric, $\sigma$ is transverse traceless with respect to $g$ (i.e., $\sigma$ is trace-free 
and divergence-free) and where we identify pairs 
\begin{equation}
(g,\sigma) \sim (\phi^{q-2}g, \phi^{-2}\sigma)
\end{equation}
for any conformal factor $\phi \in W^{k,p}_+(M)$.
\end{definition}
To complete the description of the measurement of conformal momentum, we first recall a variation of 
York splitting \cite{York:1973fl}. 
\begin{lemma}\label{lem:yorksplit}
Suppose that $A\in W^{k-1,p}(M,S_2 M)$ be trace-free and fix any $N\in W^{k,p}_+(M)$.  Then there is a unique transverse-traceless 
$\sigma\in W^{k-1,p}(M,S_2 M)$ and a vector field $W\in W^{k,p}(M,TM)$ such that
\begin{equation}
A = \sigma + \frac{1}{2N} \ck W.
\end{equation}
This formulation uses that $M$ is compact. The vector field $W$ here is uniquely determined up to addition with a 
conformal Killing field.
\end{lemma}
The special case $N\equiv 1/2$ is more commonly known as York splitting, but the result for arbitrary $N$ is a consequence 
of the $N\equiv 1/2$ case \cite{Maxwell:2014gx}, or alternatively, can be proved directly by applying Lemma \ref{lem:divck} 
below to solve
\begin{equation}
\ck^* \frac{1}{2N} \ck W = \ck^* A.
\end{equation}
\begin{definition}[Measurement of Conformal Momentum]\label{def:confmommeas}
Suppose that $\alpha$ is a fixed $W^{k,p}$ volume form on $M$.  The conformal momentum of $(g,K)$ measured by $\alpha$,
denoted $[g,K]_\alpha$, is the equivalence class of the pair $(g,\sigma)$, where $\sigma$ is computed as follows. 
Write $K = A +\frac{\tau}{n} g$ where $A$ is trace-free, and let $N = dV_g / \alpha$; then apply York splitting to decompose
\begin{equation}
A = \sigma + \frac{1}{2N} \ck W.
\end{equation}
\end{definition}

Briefly, the aim of the conformal method is to use the conformal class, conformal momentum measured by $\alpha$, and mean curvature 
as the `seed data' for solutions of the constraint equations. Fixing $\alpha$, we prescribe a conformal class $\mathbf{g}$, a 
conformal momentum $\bfsigma$, and a mean curvature $\tau$, and seek a solution $(\ol g, \ol K)$ of the vacuum constraints with
\begin{equation}\label{eq:confmethod}
[\ol g] = \mathbf g, \ \ [\ol g,\ol K]_\alpha = \bfsigma,\ \ \ol g^{ab} \ol K_{ab} 
= \tau.
\end{equation}
To cast this as a PDE, pick an arbitrary representative $g$ of $\mathbf g$, let $\sigma$ be the unique $g$-transverse-traceless 
tensor such that $(g,\sigma)$ is a representative of $\bfsigma$, and define the lapse $N=dV_g /\alpha$.  We call $(g,\sigma,\tau;\;N)$
a \textit{conformal data set}, with the lapse segregated from the other terms
to reflect its role as a gauge choice.  Starting from a conformal data set we seek a conformal factor $\phi$ and a vector field $W$ 
solving 
\begin{subequations}
\label{eq:CTS-Hvac}
\begin{alignat}{2}
-a\Delta \phi +R\phi- \left| \sigma + \frac{1}{2N}\ck W\right|^2 \phi^{-q-1} + \kappa \tau^2 \phi^{q-1}  &= 0 
&\qquad&\text{\small[CTS-H Hamiltonian constraint]}\\
\frac{1}{2}\ck^* \left[ \frac{1}{2N} \ck W \right] +\kappa \phi^{q} \extd\tau &=   0 
&\qquad&\text{\small[CTS-H momentum constraint]}
\end{alignat}
\end{subequations}
which we call the conformal thin-sandwich equations in their Hamiltonian formulation (the CTS-H equations). If $(\phi,W)$ solves these equations
then the pair 
\begin{equation}\label{eq:ctshsol}
\ol g = \phi^{q-2} g, \ \ \ol K = \phi^{-2}\left(\sigma+\frac{1}{2N}\ck W\right) + \frac{\tau}{n} \ol g
\end{equation}
solves \eqref{eq:confmethod}, and all solutions of problem \eqref{eq:confmethod} are obtained this way.

We observe that \eqref{eq:confmethod} is intrinsically conformally covariant, and hence the CTS-H equations must also be. Concretely, 
the solutions determined by $(g,\sigma,\tau;\;N)$ and
\begin{equation}\label{eq:confdatasettransform}
(\hat g,\hat\sigma,\hat\tau;\;\hat N)=(\psi^{q-2}g,\psi^{-2}\sigma,\tau;\;\psi^q N)
\end{equation}
are the same. 
In other words, we are expressing the same problem \eqref{eq:confmethod} using two different, but conformally related, 
sets of data. The standard conformal method corresponds to using the conformal representative of $\mathbf g$ 
with volume form $dV_g= \alpha/2$, so that $N\equiv1/2$ in \eqref{eq:CTS-Hvac}; we are thus 
restricting ourselves to an inflexible choice for the background metric to represent the problem, whereas if we allow an 
arbitrary background metric in the conformal class, we must introduce the lapse function $N$ into \eqref{eq:CTS-Hvac}, which
then gives 
the Hamiltonian conformal thin-sandwich method of \cite{Pfeiffer:2003ka}.

A conformal data set $(g,\sigma,\tau;\;N)$ determines the volume gauge $\alpha$ by $N = dV_g/\alpha$.  Thus, fixing
the background metric, the choice of lapse is equivalent to the choice of a volume gauge. It is important to note that 
the lapse transforms conformally by $\hat N = \psi^q N$, cf.\ \cite{YorkJr:1999jo}; we say that  
the conformal method involves a \textit{densitized lapse}.  On the other hand, the volume gauge $\alpha$ is a fixed object, and 
applies to all representatives of a conformal class. 

\section{Conformal Killing Fields and the Conformal Method}

Suppose $(\olg, K)$ is a vacuum initial data set and that $\olg$ admits a conformal Killing field $Q$.
The momentum constraint implies
\begin{equation}\label{eq:momckf}
-\ol \nabla^a(K_{ab}-\tau\;\olg_{ab}) = 0
\end{equation}
where, as usual, $\tau = \olg^{ab} K_{ab}$. Multiplying equation \eqref{eq:momckf} by the conformal Killing field $Q$,
integrating by parts, and using the conformal Killing equation
\begin{equation}
\ol \nabla_a Q_b + \ol \nabla_b Q_a = \frac{2}{n} \ol \nabla_c Q^c g_{ab},
\end{equation}
we find
\begin{equation}
\begin{aligned}
0 &= \int_M \left[ K_{ab}-\tau\;\olg_{ab}\right] \frac{1}{2}(\ol \nabla^a Q^b + \ol \nabla^b Q^a)\;dV_{\olg}\\
& = \int_M \left[ K_{ab}-\tau\;\olg_{ab}\right] \frac{1}{n} (\ol \nabla_c Q^c) \olg^{ab}\;dV_{\olg}\\
& = \frac{1-n}{n}\; \int_M  \tau\; \ol \nabla_c Q^c\;dV_{\olg}.
\end{aligned}
\end{equation}
Integrating by parts one more time gives the {\textit CKF compatibility condition}
\begin{equation}\label{eq:ckfcompat}
\int_M Q(\tau)\;dV_{\olg} = 0
\end{equation}
between mean curvature and conformal Killing fields.  

Suppose $(g,\sigma,\tau;\; N)$ is a CTS-H conformal data set where $g$ admits a conformal Killing field $Q$, and let
$(\phi,W)$ be a solution of the corresponding CTS-H equations. The CKF compatibility condition \eqref{eq:ckfcompat} then becomes
\begin{equation}\label{eq:ckfcompat-conf}
\int_M Q(\tau)\; \phi^q\; dV_{g} =0.
\end{equation}
Since \eqref{eq:ckfcompat-conf} involves the unknown $\phi$, it is not obvious whether the CKF compatibility condition imposes 
a genuine restriction on allowable conformal data sets; conceivably, the conformal method might always manage to find a 
conformal factor $\phi$ satisfying \eqref{eq:ckfcompat-conf}, regardless of the choice of $\tau$.
Nevertheless, equation \eqref{eq:ckfcompat-conf} presents an obstacle in current solution techniques for the CTS-H equations. Typically
one generates a sequence $(\phi_{(n)},W_{(n)})$ of approximate solutions iteratively; each iteration involves solving a variation of 
the momentum constraint such as
\begin{equation}\label{eq:momit}
\nabla_b  \left[\frac{1}{2N} (\ck W_{(n+1)})^{ab}\right] = \frac{n-1}{n} (\phi_{(n)})^{q}\; \nabla^a \tau.
\end{equation}
Equation \eqref{eq:momit} is solvable for $W_{(n+1)}$ if and only if
\begin{equation}\label{eq:ckfcompat-conf-it}
\int_M Q(\tau)\; \phi^q_{(n)}\; dV_{g} =0
\end{equation}
for all conformal Killing fields $Q$. Any standard method does not ensure that the successive functions $\phi_{(n)}$ still satisfy 
\eqref{eq:ckfcompat-conf-it}, so it may not be possible to continue the iteration procedure. 

Nevertheless, for certain conformal data sets, conformal Killing fields are not an obstruction to solving the CTS-H equations.  
Most importantly, if $\tau$ is constant then $Q(\tau)\equiv 0$ for any conformal Killing field, and condition \eqref{eq:ckfcompat-conf} 
is satisfied trivially for every conformal factor.  In other words, the presence of conformal Killing fields plays no role in the CMC 
theory as described in, for example, \cite{Isenberg:1995bi}. A minor generalization is that if $\tau$ is constant on the integral curves 
of every conformal Killing field $Q$, then we still have that $Q(\tau)\equiv 0$, hence \eqref{eq:ckfcompat-conf} is satisfied
trivially regardless of the conformal factor. This observation was exploited in \cite{ChoquetBruhat:1992fz} to construct near-CMC solutions 
under this hypothesis on the mean curvature function. This is a strong hypothesis, of course (and amounts to assuming that $\tau$
is constant for metrics conformal to the flat torus or the round sphere).  Moreover, this hypothesis is not necessary:
\cite{Maxwell:2011if} and \cite{Maxwell:2014Kasner} contain examples of non-CMC conformal data sets where the 
background metric is a flat torus (hence not covered by \cite{ChoquetBruhat:1992fz}) and where there exist solutions. The current
theory for the CTS-H equations does not exclude the possibility that conformal Killing fields are irrelevant to solvability.

We now give a simple example which shows that at least in certain situations, the existence theory is sensitive to the presence 
of conformal Killing fields. The argument stems from the observation that \eqref{eq:ckfcompat} is analogous to the Pohozaev constraint
\begin{equation}\label{eq:ckfcompat-R}
\int_M Q(R)\; dV_{g}=0,
\end{equation}
which relates the scalar curvature function $R$ and conformal Killing fields $Q$ on a compact manifold \cite{Bourguignon:1987ge}.
Its proof is a straightforward adaptation of ideas from \cite{KW:1974} and \cite{Bourguignon:1987ge} concerning obstructions to 
the existence of solutions for the Nirenberg problem of finding metrics in a conformal class with prescribed scalar curvature. 

\begin{proposition}\label{prop:CKFs-matter-sometimes}
Let $g$ be the round metric on the sphere $S^n$, $\sigma\not \equiv 0$ a smooth transverse traceless tensor, and $\tau_0$ a constant.
There exists a smooth function $T$ such that for every $\epsilon\in\Reals$, the conformal data set 
$(g,\sigma,\tau_0+\epsilon T; N)$ admits a solution of the vacuum CTS-H equations if and only if $\epsilon=0$.
\end{proposition}
\begin{proof}
Fix $p\in S^n$ and let $T$ be the distance function from $p$, and $Q$ the conformal Killing field $\grad T$. Define $\tau_{\epsilon}=
\tau_0+\epsilon T$. Since $(S^n,g)$ is Yamabe positive and $\sigma\not\equiv 0$, the CMC case of existence theory for the conformal 
method implies there exists a solution of the CTS-H equations when $\epsilon = 0$. On the other hand, if $\epsilon \neq 0$, then 
$Q(\tau_\epsilon) = \epsilon Q(T)$ has a single sign (except at the antipodal points), and hence $\int_{S^n} Q(\tau) \phi^q\; dV_g \neq 0$ 
for any choice of conformal factor $\phi$. This violates the CKF compatibility condition \eqref{eq:ckfcompat-conf} for every possible 
conformal factor and hence there exists no solution of the CTS-H equations for this conformal data when $\epsilon \neq 0$.
The lack of smoothness of $T$ at the antipodal points is not relevant here since we could replace $T$ by a smooth nonnegative function 
of $\mbox{dist}(p, \cdot)$ which is smooth on $S^n$ and satisfies the same conclusion. 
\end{proof}

Proposition \ref{prop:CKFs-matter-sometimes} shows that there exist CMC solutions of the constraint equations such that,
replacing the mean curvature function by certain arbitrarily small perturbations of it in the conformal data set, 
then the CTS-H equations no longer have a solution. This means that the standard hypothesis in the near-CMC theory that 
the metric does not admit nontrivial conformal Killing fields cannot be dropped completely.  
It is not at all clear if there is some natural and easily apparent geometric condition that distinguishes when one
should expect there to exist solutions or not.  

We shall take an alternate course and give up on prescribing the mean curvature function specifically.  The drift method
described in the next section involves the prescription of different sets of data, and implicitly shows how to adjust 
the mean curvature to account for the CKF compatibility condition. 

\section{Drift Variations of the Conformal Method}
In this section we give a brisk description of the drift formulations of the conformal method \cite{Maxwell:2014Drift}. 
Before getting into details, we observe that the principal distinction between the drift and standard conformal methods 
is that while the conformal method prescribes the mean curvature $\tau$ of the solution directly, the drift techniques 
involve a decomposition 
\begin{equation}
\tau = \tau_* + \frac{1}{N} \div V
\end{equation}
where $\tau_* \in \mathbb R$, $N$ is a positive function (the same lapse appearing in the CTS-H equations) and $V$ is
a vector field.  The mean curvature determined by $\tau_*$ and $V$ changes as we move between representatives in
a given conformal class for several reasons. First, the divergence operator depends on the choice of representative. 
Second, the lapse transforms as a densitized lapse, as described at the end of Section \ref{sec:confmeth}.  
Finally, if the metric admits conformal Killing fields, we cannot prescribe $V$ directly, but must add a suitable 
conformal Killing field $Q$ that changes as we change the conformal class representative. So in general,
\begin{equation}\label{eq:taudecomp}
\tau = \tau_* + \frac{1}{N} \div (V+Q)
\end{equation}
where $Q$ is a conformal Killing field determined by $V$. In short, the actual mean curvature function determined by
the data in the drift formulations naturally adapts to the presence of conformal Killing fields. This allows one
to prove slightly more general results. 

As discussed next in Section \ref{sec:volmom}, the constant $\tau_*$ in equation \eqref{eq:taudecomp} represents a certain 
dynamical quantity called the volumetric momentum. To interpret the vector field $V$, we first note that the mean curvature is
unchanged by adding a divergence-free vector field to $V$.  Moreover, 
$V$ is prescribed only up to adjustment by a suitable conformal Killing field, so $V$ is an element of the space
of vector fields modulo both conformal Killing fields and divergence-free vector fields. This quotient space is the space
of so-called drifts. We discuss them further in Section \ref{sec:driftdef}, before giving the equations for the
drift formulations in Section \ref{sec:driftform}.

\subsection{Volumetric Momentum}\label{sec:volmom}
The first step toward the drift parameterization of the constraint equations involves the identification
of a parameter, volumetric momentum. This plays a role somewhat analogous to the transverse-traceless
tensor in the standard conformal method, and represents a cotangent vector to the one-dimensional 
space of volume forms modulo diffeomorphisms, so the volumetric momentum is just a number. It arises in the following
analog of York splitting.
\begin{lemma}\label{lem:yorksplitvol}
Let $\tau\in W^{k-1,p}(M)$ and let $N\in W^{k,p}_+(M)$.  There is a unique constant $\tau_*$
and a vector field $V\in W^{k,p}$ such that
\begin{equation}\label{eq:yorksplitvol}
\tau = \tau_* + \frac{1}{N} \div V.
\end{equation}
Moreover, $V$ is uniquely determined up to addition of a divergence-free vector field, and 
\begin{equation}
\tau_* = \frac{\int_M N \tau\; dV_g}{\int_M N\; dV_g}.
\end{equation}
\end{lemma}
This is proved in \cite{Maxwell:2014Drift} when the data is smooth, but the same proof works for metrics and data 
with the regularity stated here. 

\begin{definition}[Measurement of Volumetric Momentum]\label{def:volmommeas}
Let $\alpha$ be a $W^{k,p}$ volume form on $M$.  The \textit{volumetric momentum} of $(g,K)$ measured by $\alpha$
is computed as follows.  First, let $N = dV_g / \alpha$ and define $\tau = g^{ab}K_{ab}$. By Lemma \ref{lem:yorksplitvol},
$\tau  = \tau_* + \frac{1}{N} \div V$
for a unique constant $\tau_*$.  The volumetric momentum of $(g,K)$ measured by $\alpha$ is 
\begin{equation}
[g,\tau]_\alpha = -2\kappa \tau_*, \ \ \kappa = (n-1)/n.
\end{equation}
\end{definition}

Volumetric momentum is already an interesting parameter in the standard conformal method.  Examples in 
\cite{Maxwell:2014Kasner} exhibit the development of certain one-parameter families of non-CMC solutions 
of the constraint equations generated by the standard conformal method, and $\tau_*=0$ is among the 
several necessary conditions needed to generate these families.  Curiously, $\tau_*=0$ is not easily detected 
from the usual conformal data; in effect, one must solve the equations of the conformal method to determine 
if $\tau_*$ vanishes or not. These examples motivate finding a parameterization in which $\tau_*$ is 
explicitly prescribed. 

\subsection{Drift}\label{sec:driftdef}
Fixing a volume gauge $\alpha$, the conformal momentum and volumetric momentum 
of $(g,K)$ measured by $\alpha$ drop out of the momentum constraint.
Indeed, using Lemmas \ref{lem:yorksplit} and \ref{lem:yorksplitvol} to decompose
\begin{equation}
K = \sigma + \frac{1}{2N} \ck W + \frac{1}{n}\left[\tau_* + \frac{1}{N}\div V\right] g,
\end{equation}
then the vacuum momentum constraint becomes
\begin{equation}\label{eq:drift}
-\frac{1}{2}\ck^* \frac{1}{2N} \ck W  = \kappa \extd \left[\frac{1}{N}\div V\right].
\end{equation}
The momentum equation in this formulation has interesting symmetries.
The vector fields $W$ and $V$ appearing in it each represent a certain geometric object,
coined a drift in \cite{Maxwell:2014Drift}.
\begin{definition}[Drift]\label{def:drift}
Let $g$ be a $W^{k,p}$ metric.  A \textit{drift} at $g$ is an element of
\begin{equation}
W^{k,p}(M,TM)/ (\Ker \ck_g + \Ker \div_g).
\end{equation}
We write $\drift[g]{W}$ for the drift at $g$ determined by the vector field $W$
and $\Drift_g$ for the space of drifts at $g$.
\end{definition}
\begin{remark} The spaces $\Ker \ck_g$ and $\Ker \div_g$ intersect in the space of Killing fields,
but since $\Ker \div_g$ is closed in $W^{k,p}(M,TM)$ and $\Ker \ck_g$ is finite dimensional, $\Ker \ck_g + \Ker \div_g$ 
is also a closed subspace and the quotient $\Drift_g$ inherits a Banach space topology. 
\end{remark}
As elaborated in \cite{Maxwell:2014Drift}, a drift represents an infinitesimal motion 
in the space of metrics, modulo diffeomorphisms, that preserves the conformal class up to 
diffeomorphism and the volume. Note that such a motion need not preserve the diffeomorphism class of 
the metric. 

Equation \eqref{eq:drift} represents a relationship between two drifts. To see this, suppose $\mathbf V$ is a drift 
at $g$, with $V$ any representative.  Equation \eqref{eq:drift} can be regarded as a PDE in $W$. If $g$ admits 
conformal Killing fields, there is no solution unless the right-hand side of \eqref{eq:drift} is orthogonal to $\ck_g$.
Assuming this orthogonality, hypothesis Theorem 10.1 of \cite{Maxwell:2014Drift} shows (in the smooth category)
that there is a conformal Killing field $Q$ and a vector field $W$ such that 
\begin{equation}
-\frac{1}{2}\ck^* \frac{1}{2N} \ck W  = \kappa \extd \left[\frac{1}{N}\div (V+Q)\right].
\end{equation}
Here $Q$ is uniquely determined up to a true Killing field, $W$ is uniquely determined up to a conformal Killing field, 
and $\drift[g]{W}$ is independent of the choice of representative of the drift $\mathbf V$.

This process can be reversed. Suppose $\mathbf W \in \Drift_g$ and let $W$ be any representative.  We now wish to
solve \eqref{eq:drift} for $V$, but to do so, 
the left-hand side of \eqref{eq:drift} must be orthogonal to the space of divergence-free vector fields.
Theorem 10.6 of \cite{Maxwell:2014Drift} shows (in the smooth category)
that there is a divergence-free vector field $E$ and a vector field $V$ such that 
\begin{equation}
-\frac{1}{2}\ck^* \frac{1}{2N} \ck (W+E)  = \kappa \extd \left[\frac{1}{N}\div (V)\right].
\end{equation}
Now $E$ is uniquely determined up to a true Killing field, $V$ is
uniquely determined up to a divergence-free vector field, and $\drift[g]{V}$
is independent of the choice of representative of the drift $\mathbf W$.

Motivated by this discussion, we assign a pair of drifts to a pair $(g,K)$ as follows.

\begin{definition}[Measurement of Drift]\label{def:driftmeas}
Suppose $g$ is a $W^{k,p}$ metric, $K\in W^{k-1,p}(M,S_2M)$ and $\alpha$ is a $W^{k,p}$ volume form. 
Set $N=dV_g/\alpha$ and decompose
\begin{equation}
K = A + \frac{\tau}{n} g,
\end{equation}
where $A$ is trace-free. Now use Lemmas \ref{lem:yorksplit} and \ref{lem:yorksplitvol} to write
\begin{equation}
A = \sigma + \frac{1}{2N} \ck W
\end{equation}
and
\begin{equation}
\tau = \tau_* + \frac{1}{N} \div V.
\end{equation}
The \textit{volumetric drift} of $(g,K)$ measured by $\alpha$ is $\drift[g]{V}$, and
the \textit{conformal drift} measured by $\alpha$ is $\drift[g]{W}$.
\end{definition}

In our application of drifts to the construction of near-CMC solutions of the Einstein constraint 
equations we shall specify the conformal class of the solution metric and, among other parameters, 
the volumetric drift. Since drift is defined in terms of a metric rather than the conformal
class, one needs to be able to specify a drift at an unknown solution metric $\overline g = \phi^{q-2}g$ 
starting from a given representative $g$ of the conformal class.  One can always specify the vector field 
$V$ and let it determine the drift $\drift[\overline g]{V}$, but unless one knows the conformal
factor $\phi$, it is impossible to know \textit{a priori} whether $V$ is divergence-free with respect 
to the solution metric and hence $\drift[\overline g]{V}=0$.  

To address this difficulty, suppose for the moment that the $W^{k,p}$ metric $g$ admits no (nontrivial) 
conformal Killing fields. The Helmholtz decomposition implies
\[
W^{k,p}(M,TM) = \mathcal E \oplus \mathcal E^\perp
\]
where $\mathcal E$ is the set of $W^{k,p}$ divergence-free vector 
fields and $\mathcal E^\perp$ is the image of $\grad$ acting 
on $W^{k+1,p}$ functions. The factors in the direct sum are $L^2$ orthogonal,
and the projection of a vector field $X$ onto $\mathcal E^\perp$
is $\grad u$ where $\Lap u = \div X$.  Because $g$  admits
no conformal Killing fields, the drifts at $g$ can be
identified with $\mathcal E^\perp$.  Moreover, for 
a conformally related metric $\overline g=\phi^{q-2} g$
the conformal transformation rule for gradients implies
\[
\mathcal E^\perp_{\overline g} = \phi^{2-q} \mathcal E^\perp_{g}.
\]
Hence in absence of conformal Killing fields
we have a mechanism for parameterizing drifts within a 
conformal class: drifts can be represented by elements 
of $\mathcal E^\perp_g$
and $V\in \mathcal E^\perp_{g}$ 
corresponds to $\phi^{2-q}V\in \mathcal E^\perp_{\overline g}$.

In the event that $g$ admits conformal Killing fields the
representation of drift within a conformal class is less 
straightforward because divergence-free vector fields and conformal
Killing fields obey different conformal transformation laws.  
In this case the drifts at $g$ can be identified with any one
of a number of subspaces of $\mathcal E^\perp$, and it seems natural to use the $L^2$ orthogonal 
complement of $P(\mathcal Q)$, where $P$ is 
the $L^2$ projection of $W^{k,p}(M,TM)$ 
onto $\mathcal E^\perp$ discussed above. 

\begin{definition} A \textit{canonical drift representative} 
at a $W^{k,p}$ metric $g$
is a vector field $V\in \mathcal E^\perp$ satisfying
\[
\int_M g(V,P(Q)) \;dV_g = 0
\]
for all conformal Killing fields $Q$.  
The set of canonical drift representatives at $g$
is denoted by $\mathcal D_g$.
\end{definition}

It is easy to see that the map $V \mapsto \drift[g]{V}$
from $\mathcal D_g$ to 
$\Drift_g$ 
is a Banach space isomorphism and hence $\Drift_g$ can be identified with
a subspace of $\mathcal E_g^\perp$ with codimension
equal to $\dim P(\mathcal Q)$. This codimension need not 
be constant among all representatives
of a conformal class, however.  Indeed,  
a conformal Killing field $Q$ is 
a true Killing field exactly
when it is divergence-free, i.e. when $P(Q)=0$. 
Thus
$\dim P(\mathcal Q) \le \dim \mathcal Q$ with strict inequality
whenever the metric admits nontrivial Killing fields. 
The non-constant codimension of $\mathcal D_g$ in $\mathcal E_g^\perp$
poses an obstacle 
to the universal representation of drift for a fixed conformal class. Nevertheless, our main application of
drifts to the conformal method is perturbative, and
the following lemma shows that we can use $\mathcal D_g$
to identify drifts at nearby representatives of the conformal class 
so long as $g$ does not admit any true Killing fields.

\begin{lemma}\label{lem:driftrep}
Let $g$ be a $W^{k,p}$ metric.
Given a conformal factor $\phi\in W^{k,p}$ let $\overline g=\phi^{q-2}g$.
The map 
\[
V\mapsto \drift[\overline g]{\phi^{2-q}V}
\]
from $\mathcal D_{g}$ to $\Drift_{\overline g}$
is an isomorphism if either
\begin{itemize}
\item $g$ admits no (nontrivial) conformal Killing fields, or
\item $g$ admits no (nontrivial) Killing fields and $\phi$ is sufficiently
close to $1$ in $W^{k,p}$.
\end{itemize}
\end{lemma}
\begin{proof}
Let $\mathcal E_\phi$, $\mathcal E_\phi^\perp$, $\mathcal D_\phi$
and so forth represent objects associated with the metric 
$\overline g=\phi^{q-2}g$, let $P_\phi$ be the $\ol g$-$L^2$ projection
of $W^{k,p}(M,TM)$ onto $\mathcal E_\phi^\perp$, 
and let $D_\phi$ be the $\ol g$-$L^2$ projection 
of $W^{k,p}(M,TM)$ onto $\mathcal D_\phi$.  One
readily verifies that these maps are continuous, in part
using the
standing hypotheses on $k$ and $p$ (which ensure that $W^{k,p} \subset L^2$), along with the fact that $P(\mathcal Q)$ is finite dimensional.
Given a vector field $V$, the projections $P_\phi(V)$ and $(D_\phi\circ P_\phi)(V)$ differ from $V$ by linear combinations of conformal Killing fields
and $\overline g$-divergence-free vector fields.  Hence
\[
\drift[\overline g]{V} = \drift[\overline g]{P_\phi V} = 
\drift[\overline g]{(D_\phi\circ P_\phi) (V)}.
\]
In particular, if $V\in \mathcal D_1$, then $\phi^{2-q} V\in \mathcal E_\phi^\perp$ and 
\[
\drift[\overline g]{\phi^{2-q}V} = \drift[\overline g]{D_\phi(\phi^{2-q}V)}.
\]
Since the projection from $\mathcal D_{\overline g}$ 
onto $\Drift_{\overline g}$ is an isomorphism it is therefore enough to
show that $F_\phi:\mathcal D_1 \rightarrow \mathcal D_\phi$
defined by
\[
F_\phi(V) = D_\phi( \phi^{q-2} V)
\]
is an isomorphism under the given hypotheses on $g$ and $\phi$. 

Suppose first that $g$ admits no conformal Killing fields. 
In this case $\mathcal D_1 = \mathcal E_1^\perp$,
$\mathcal D_\phi = \mathcal E_\phi^\perp$,  and the result
follows from the previously discussed isomorphism 
$\mathcal E_1^\perp\rightarrow \phi^{2-q}\mathcal E_1^\perp=\mathcal E_\phi^\perp$.

Now consider
the case where $g$ admits conformal Killing fields,
but no Killing fields.  Let $G_\phi$  be the 
$\ol g$-$L^2$ projection of $W^{k,p}(M,TM)$
onto $P_\phi(\mathcal Q)$; we claim that $G_\phi$
is continuous in $\phi$ when thought of as a map
with codomain $W^{k,p}(M,TM)$.  Indeed first
note that the maps $P_\phi$, defined previously in terms of
solving a Poisson problem for the metric $\ol g$,
are continuous in $\phi$. Hence,
fixing a basis $\{Q_i\}$ for $\mathcal Q$,
the vectors $P_\phi(Q_i)$ also depend continuously
on $\phi$.  Moreover, since $g$ has no Killing fields
the map $P_1|_{\mathcal Q}$ is injective, 
and the continuity of $P_\phi$ with respect to
$\phi$ (along with the fact that $\mathcal Q$ is finite
dimensional) ensures that 
$P_\phi|_{\mathcal Q}$ is injective for $\phi$
sufficiently close to 1. Hence the vectors
$\{P_\phi(Q_i)\}$ are linearly independent.
The map taking a frame (in this case $\{P_\phi(Q_i)\}$) to 
an orthonormal frame via the Gram-Schmidt algorithm is continuous in
the frame and the inner product jointly, and 
writing the projection $G_\phi$ with respect
to the orthonormal frame it readily follows 
that $G_\phi$ is continuous
in $\phi$, as is $D_\phi = \rm{Id}-G_\phi$ with
codomain $W^{k,p}(M,TM)$.

Now consider the maps 
\[
B_\phi = (D_1\circ S_\phi^{-1})\circ (D_\phi\circ S_\phi):\mathcal D_1 \rightarrow \mathcal D_1
\]
where $S_\phi(V) = \phi^{2-q} V$.  From our observation
that $D_\phi$ is continuous in $\phi$, so are the maps $B_\phi$.
And since $B_1=\mathrm{Id}$, we conclude that $B_\phi$ is 
an isomorphism for $\phi$ sufficiently close to 1.
Noting that
\begin{equation}\label{eq:Bfactorize}
B_\phi = (D_1\circ S_\phi^{-1}|_{\mathcal D_\phi})\circ F_\phi,
\end{equation}
to show that $F_\phi$ 
is an isomorphism
for $\phi$ close to $1$ it is therefore enough to establish
the same fact for 
$D_1\circ S_\phi^{-1}|_{\mathcal D_\phi}:\mathcal D_\phi \rightarrow \mathcal D_1$.  Moreover, the factorization 
\eqref{eq:Bfactorize} already implies that 
$D_1\circ S_\phi^{-1}|_{\mathcal D_\phi}$ is surjective for all conformal
factors sufficiently near $1$, and we need only 
establish injectivity.  

The kernel of $D_1$ is $P_1(\mathcal Q)$
and hence 
\[
\Ker D_1\circ S_\phi^{-1}|_{\mathcal D_\phi} = (S_\phi\circ P_1)(\mathcal Q)
\cap \mathcal D_\phi.
\]
Now $(S_\phi\circ P_1)(\mathcal Q)\subseteq \mathcal E_\phi^\perp$, and since
$\Ker G_\phi|_{\mathcal E_\phi^\perp} = \mathcal D_\phi$,
to show that the subspace $(S_\phi\circ P_1)(\mathcal Q)
\cap \mathcal D_\phi$ is trivial (for $\phi$ close to 1) it is enough to show that
\[
G_\phi \circ S_\phi\circ P_1|_{\mathcal Q}:\mathcal Q\rightarrow W^{k,p}(M,TM)
\]
is injective. 
But this follows from the fact that
this family of linear maps has finite-dimensional domain,
is continuous in $\phi$, and is injective at $\phi\equiv 1$.
\end{proof}

\subsection{Parametrizations of the Constraints Using Conformal Deformation, Expansion and Drift}
\label{sec:driftform}
Recall that in the standard conformal method we prescribe
the conformal class, conformal momentum, and mean curvature of the solution.
In the drift formulation, we replace the mean curvature
with the combination of volumetric momentum and either volumetric or conformal drift.

Consider a solution $(\overline g, K)$ of the vacuum constraint
equations, and let $\alpha$ be a volume gauge.  From Lemmas
\ref{lem:yorksplit} and \ref{lem:yorksplitvol}
the solution uniquely determines
\begin{itemize} 
	\item a conformal class $[\overline g]$,
	\item a conformal momentum measured by $\alpha$
	represented by $(\overline g, \overline \sigma)$
	where $\overline \sigma$ is transverse-traceless with
	respect to $\overline g$,
	\item a volumetric momentum $-2\kappa \tau_*$ measured by $\alpha$
	\item and a volumetric drift $\drift[\overline g]{\overline V}$ measured by $\alpha$.
\end{itemize}
The first three parameters can be prescribed in a conformally invariant
fashion by choosing a representative $g$
of the conformal class, along with a $g$-transverse-traceless
tensor $\sigma$ and a constant $\tau_*$.  Then
$\overline g=\phi^{q-2}g$ for some conformal factor
$\phi$ and $\overline \sigma=\phi^{-2} \sigma$.  As
for the drift, recall from Lemma \ref{lem:driftrep} 
that the map $\mathcal D_g \rightarrow \Drift_{\overline g}$
given by $V \mapsto \drift[\overline g]{\phi^{2-q}V}$ is an isomorphism
so long as $g$ has no conformal Killing fields, or so
long as $g$ has no Killing fields and $\overline g$
is sufficiently close to $g$.  Hence we will select $V\in\mathcal D_g$
and set $\overline V = \phi^{2-q}V$, and the aim of the drift
method is to recover the solution of the constraint equation from
these parameters.  

More precisely, we prescribe the following conformal data:
\begin{equation}
(g, \sigma, \tau_*, V;\; N)
\end{equation}
where $\sigma$ is transverse-traceless, $\tau_*$ is a constant, 
$V\in \mathcal D_g$ is a canonical drift representative at $g$, 
and $N$ is lapse 
specifying a volume gauge $\alpha$ according to the relationship $N=dV_g/\alpha$.
We seek a solution $(\phi,W,Q)$ of the following 
variation of equations 
from \cite{Maxwell:2014Drift} Section 12: 
\begin{equation}\label{eq:driftV-vac}
\begin{aligned}
-a\Lap \phi + R \phi + \left|\sigma + \frac{1}{2N}\ck W \right|^2 \phi^{-q-1} 
+ \kappa \left(\tau_* + \frac{1}{N\phi^{q}}\div_\phi(\phi^{2-q}V + Q)) \right)^2 \phi^{q-1} &= 0\\
\frac{1}{2} \ck^* \left(\frac{1}{2N}\ck W\right) -\kappa \div_\phi^* \left( \frac{1}{N} \div_\phi (\phi^{2-q}V+Q) \right) &= 0
\end{aligned}
\end{equation}
where $W$ is an arbitrary vector field and $Q$ is a conformal Killing field.  Here we are using
the notation
\begin{equation}
\div_\phi = \phi^{-q} \div \phi^{q}
\end{equation}
for the divergence operator of the metric $\phi^{q-2}g$,
while
\begin{equation}
\div_\phi^* = -\phi^{q}\; \extd\; \phi^{-q}
\end{equation}
is the adjoint of $\div_\phi$ with respect to the background metric $g$.
The conformal Killing field $Q$ is determined by the CKF compatibility condition
\begin{equation}\label{eq:ckfcompat-driftV-vac}
\int \frac{1}{N} \div_\phi (\phi^{2-q}V+Q) \div_\phi P\;  dV_g  = 0
\end{equation}
for all conformal Killing fields $P$, which can be added to system
\eqref{eq:driftV-vac} to make the number of equations match the number of unknowns.

Supposing $(\phi,W,Q)$ solves system \eqref{eq:driftV-vac}, let 
\begin{equation}
\begin{aligned}
\label{eq:CEDV-reconstruct}
\ol g &= \phi^{q-2} g\\
\tau &= \tau_* + \frac{1}{N\phi^q} \div_\phi (\phi^{2-q}V+Q)\\
\ol K &=  \phi^{-2}\left[ \sigma +\frac{1}{2N}\ck W\right] +\frac{\tau}{n}\ol g.
\end{aligned}
\end{equation}
Following arguments from \cite{Maxwell:2014Drift} it follows 
that $(\ol g,\ol K)$ is a solution of the constraints
with conformal class $[g]$, conformal momentum represented by $(g,\sigma)$, 
volumetric momentum $-2\kappa\tau_*$, and volumetric drift $\drift[\overline g]{\phi^{2-q}V}$ as desired. We will
call system \eqref{eq:driftV-vac} together with \eqref{eq:ckfcompat-driftV-vac} 
the vacuum CED-V equations, short for conformal deformation, expansion, and (volumetric) drift.  

Alternatively, we can prescribe the conformal drift instead of the volumetric drift.
Starting with conformal data
\begin{equation}
(g, \sigma, \tau_*, W;\;N)
\end{equation}
with $W\in\mathcal D_{g}$ we seek a solution $(\phi,V,E)$ of
\begin{equation}\label{eq:driftC-vac}
\begin{aligned}
-a\Lap \phi + R \phi + \left|\sigma + \frac{1}{2N}\ck (\phi^{2-q}W + \phi^{-q}E) \right|^2 \phi^{-q-1} 
+ \kappa \left(\tau_* + \frac{1}{N\phi^{q}}\div_\phi(V + Q)) \right)^2 \phi^{q-1} &= 0\\
\frac{1}{2} \ck^* \left(\frac{1}{2N}\ck (\phi^{2-q}W + \phi^{-q}E)\right) -\kappa \div_\phi^* \left( \frac{1}{N} \div_\phi (V) \right) &= 0
\end{aligned}
\end{equation}
where $V$ is an arbitrary vector field and $E$ is divergence free.  The vector field $E$
is determined by the compatibility condition
\begin{equation}
-\int \frac{1}{4N}\ck (\phi^{2-q}W + \phi^{-q}E)\ck(\phi^{-q}F)\;dV_g = 0
\end{equation}
for all divergence-free vector fields $F$.
Given $(\phi,V,E)$ solving system \eqref{eq:driftC-vac}, let 
\begin{equation}
\begin{aligned}
\ol g &= \phi^{q-2} g\\
\tau &= \tau_* + \frac{1}{N\phi^q} \div_\phi (V)\\
\ol K &=  \phi^{-2}\left[ \sigma +\frac{1}{2N}(\ck W+\phi^{-q}E)\right] +\frac{\tau}{n}\ol g.
\end{aligned}
\end{equation}
We find $(\ol g,\ol K)$ is a solution of the constraints as before,
except that we have prescribed conformal drift $\drift[\overline g]{\phi^{2-q}W}$ rather than volumetric drift,
and we will call system \eqref{eq:driftC-vac} the vacuum CED-C equations.

\section{Conformal Description of Matter}\label{sec:matter}
Section \ref{sec:confmeth} described the conformal method in terms of natural geometric parameters such as 
conformal momentum. By contrast, the current literature for including matter in the conformal method is 
somewhat ad hoc, and is guided by finding formulations that make the problem mathematically 
tractable \cite{ChoquetBruhat:2009hv}.  We note, for example, the methods of scaling and unscaling sources,
in the vocabulary of \cite{ChoquetBruhat:2000iw}. It has long been understood that in the CMC 
case the conformal method is compatible with scaling sources, whereas unscaling sources lead to undesirable 
non-uniqueness properties \cite{Pfeiffer:2005iz}.  We also point to \cite{Isenberg:1977hy}, which enunciates
a fundamental guiding principle that leads to to the method of scaling sources; in effect we specify the 
configuration and momentum of matter independent of the metric.\footnote{In light of \cite{Isenberg:1977hy},
the term `scaling sources' is a misnomer.  In the method of scaling sources 
the configuration of matter is conformally invariant, and only the metric used to measure
it changes.} Given our interest in constructing near-CMC solutions, we employ scaling sources in the 
framework laid out in \cite{Isenberg:2005we}. This is described briefly here without any focus on 
the underlying principle of \cite{Isenberg:1977hy}.

We represent matter fields as sections $\mathcal F$ of a smooth vector bundle over $M$. 
The energy and momentum densities of the matter fields are functions jointly of $\mathcal F$ and the 
metric $g$,
\begin{equation}
\mathcal E(\mathcal F,g)\ \mbox{and} \ \mathcal J(\mathcal F,g)
\end{equation}
respectively, and with this notation the full
Einstein constraint equations read
\begin{subequations}\label{eq:constraints-matter}
\begin{alignat}{2}
R_{ g} - | K|_{ g}^2 + (\tr_{ g}  K)^2 &= 16\pi \mathcal E(\mathcal F,g) + 2\Lambda\\
\div_{ g}  K- \extd (\tr_{g} K) &= -J(\mathcal F,g)
\end{alignat}
\end{subequations}
where $\Lambda$ is the cosmological constant.

We assume that $\mathcal F$ obeys a conformal transformation law.
Specifically, if the metric changes from $g$ to $\hat g = \phi^{q-2}g$ then
the fields transform according to $\widehat{\mathcal{F}} = \Phi(\mathcal{F},\phi)$
where $\Phi$ is a group action of the conformal factors on the matter fields, i.e., $\Phi(\mathcal F,1)=\mathcal F$  and $\Phi(\Phi(\mathcal F,\phi_1),\phi_2)
=\Phi(\mathcal F,\phi_1\phi_2)$.  We assume moreover that any necessary 
compatibility conditions on the matter fields (e.g. the divergence-free condition
for magnetic fields) are preserved as we transform
from $g$ to $\hat g$ and $\mathcal F$ to $\widehat{ \mathcal F}$.
The key hypothesis for scaling sources is that
\begin{equation}\label{eq:momconftrans}
\mathcal J(\Phi(\mathcal F,\phi),\phi^{q-2}g) = \phi^{-q} \mathcal J(\mathcal F,g).
\end{equation}
This perhaps unmotivated scaling occurs naturally in practice and for CMC conformal data leads to a momentum 
constraint that is semi-decoupled from the Hamiltonian constraint.
Fixing $\calF$ at the metric $g$, the transformation law \eqref{eq:momconftrans}
amounts to assuming that the momentum
density is described by a one form $j$ that conformally transforms according to
\begin{equation}\label{eq:momentumtrans}
\hat j = \psi^{-q} j.
\end{equation}

Turning to the energy density, again fix $\mathcal F$ at $g$ and
define
\begin{equation}
\rho(\phi) = \mathcal E(\Phi(\mathcal F,\phi),\phi^{q-2}g).
\end{equation}
The details of this map 
depend strongly on the specific type of matter, and we make the following minimal hypothesis.
\begin{definition}\label{def:rhoscale}
A smooth map $\rho:W^{k,p}_+(M)\ra W^{k-2,p}(M)$ satisfies the 
\textit{energy scaling condition} if: 
\begin{enumerate}
\item
The linearization of $\rho$ at $\phi$ in the direction $\dot\phi$ can be written in the form
\begin{equation}
D\rho_\phi[\dot\phi] = r \dot\phi
\end{equation}
where $r \in W^{k-2,p}(M)$ depends on $\phi$.
\item Either 
\begin{itemize}
\item $\rho(\phi) \equiv 0$ for all $\phi\in W^{k,p}_+(M)$, or 
\item for all $\phi\in W^{k,p}_+(M)$
the $W^{k-2,p}(M)$ function that is the linearization of
\begin{equation}
\phi\mapsto \phi^{q-2}\rho(\phi)
\end{equation}
is non-positive and not identically zero.
\end{itemize}
\end{enumerate}
\end{definition}
As with hypothesis \eqref{eq:momentumtrans} for the momentum density, 
Definition \ref{def:rhoscale} is somewhat unmotivated, 
but admits the following
loose interpretation: energy density measured by the metric is a local property, depending on the value 
of $\phi$ but not its derivatives, and it grows at least as fast as $\phi^{2-q}$ as $\phi\ra 0$,
and decays at least as fast as  $\phi^{2-q}$ as $\phi\ra\infty$.
We will use the notation $\rho(\cdot)$ for the map $\rho$ as a reminder that it is a
function taking a conformal factor as an argument, rather than simply a function defined on $M$.

The framework for matter described here is broad enough to include a number
of important matter models including electromagnetism (and Yang-Mills fields more generally), 
perfect fluids (including dust), and Vlasov models.  These details were treated in \cite{Isenberg:2005we}, where the
energy scaling condition appears in a somewhat obscured form as hypothesis 
N1.\footnote{Condition N1 of \cite{Isenberg:2005we} is equivalent to
\begin{equation}
\phi\mapsto \phi^{q-2}\rho(\phi)
\end{equation}
being decreasing in $\phi$. The hypothesis
of Definition \ref{def:rhoscale} that it is \textit{strictly} 
decreasing somewhere (except in vacuum) does
not appear explicitly in \cite{Isenberg:2005we}, but is easy enough to verify from the expressions computed in that paper that this
additional condition is satisfied for all the specific matter
fields treated in that work.}
This framework notably excludes scalar fields, however, where condition 2 of Definition \ref{def:rhoscale} fails, and we refer
to \cite{Hebey:2008gk} for alternate techniques needed to include scalar fields in the conformal method.
As a concrete example, consider electromagnetism in $3$-dimensions without charged sources.  
The matter fields consist of divergence-free one-forms $E$ and $B$ representing the
electric and magnetic fields and we have energy and momentum densities
\begin{equation}
\begin{aligned}
\calE(E,B,g) &= |E|^2_g+|B^2|_g^2\\
\calJ(E,B,g) &= *_g(E\wedge B)
\end{aligned}
\end{equation}
where $*_g$ is the Hodge-star operator.  We conformally transform the fields according to
$\Phi((E,B),\phi)=(\phi^{-2}E,\phi^{-2} B)$ which preserve the conditions that
these one-forms must be divergence-free.  One readily verifies that
\begin{equation}
\begin{aligned}
\calJ(\phi^{-2}E,\phi^{-2}B,\phi^{q-2}g) &= \phi^{-2}*_g(\phi^{-2}E\wedge \phi^{-2} B)\\
&= \phi^{-q} \calJ(E,B,g)
\end{aligned}
\end{equation}
since $q=6$ when $n=3$. Thus we can take $j = *_g(E\wedge B)$.  For the Hamiltonian constraint we have
\begin{equation}
\begin{aligned}
\calE(\phi^{-2}E,\phi^{-2}B,\phi^{q-2}g) &= \phi^{-8}|E|^2+\phi^{-8}|B|^2
\end{aligned}
\end{equation}
and hence
\begin{equation}
\rho(\phi) = \left[|E|^2 + |B|^2\right]\phi^{-8}.
\end{equation}
Noting that $q-2=4$ when $n=3$, 
\begin{equation}
\phi^{q-2}\rho(\phi) = \left[|E|^2 + |B|^2\right]\phi^{-4}
\end{equation}
which evidently satisfies the energy scaling condition.

For convenience, we treat the cosmological constant $\Lambda$ as an additional form of matter,
and we will call a triple $(\rho(\cdot),j,\Lambda)$ where $\rho(\cdot)$ satisfies
the conditions of Definition \ref{def:rhoscale} a \textit{conformal matter distribution}.
The CTS-H equations include a conformal matter distribution according to
\begin{equation}
\begin{aligned}\label{eq:CTS-H}
-a\Delta \phi +R\phi- \left| \sigma + \frac{1}{2N}\ck W\right|^2 \phi^{-q-1} + \kappa \tau^2 \phi^{q-1}  &= 2\left[ 8\pi \rho(\phi) \phi^{q-1} + \Lambda\phi^{q-1}\right]\\
\frac{1}{2}\ck^* \left[ \frac{1}{2N} \ck W \right] + \kappa \phi^{q} d\tau &=  8\pi j.
\end{aligned}
\end{equation}
If $(\phi,W)$ is a solution of these equations then $(\ol g,\ol K)$ defined
by equations \eqref{eq:ctshsol} solve the constraint equations for matter fields
$\ol{\mathcal F} = \Phi(\calF,\phi)$  giving an energy density $\ol\rho = \rho(\phi)$ and
a momentum density $\ol j = \phi^{-q} j$. We will call $(\ol\rho,\ol j,\Lambda)$
a physical matter distribution. An easy computation using the fact that 
$\Phi$ is group action shows that the CTS-H equations with matter are conformally
covariant as well, so long as when we conformally transform
to $\hat g=\psi^{q-2}g$ we also transform to the field $\hat{\mathcal F} = \Phi(\calF,\hat g)$
to obtain $\hat\rho(\cdot) = \rho(\psi\; \cdot)$ and $\hat j = \psi^{-q}j$.

Because the drift formulations differ from the CTS-H equations only in their treatment of the
mean curvature, a conformal matter distribution $(\rho(\cdot),j,\Lambda)$ appears in the drift formulations
of the constraint equations in exactly the same way as for the CTS-H equations.  Simply replace the
zeros on the right-hand sides of equations \eqref{eq:driftV-vac} or \eqref{eq:driftC-vac}
with the right-hand sides of equations \eqref{eq:CTS-H}, but observe
that the associated compatibility conditions need to account for the momentum density. 

The CED-V equations in their final form, extending system \eqref{eq:driftV-vac} to include matter, are
\begin{equation}\label{eq:driftV2}
\begin{aligned}
-a\Lap \phi + R \phi & + \left|\sigma + \frac{1}{2N}\ck W \right|^2 \phi^{-q-1} 
+ \kappa \left(\tau_* + \frac{1}{N\phi^{q}}\div_\phi(\phi^{2-q}V + Q)) \right)^2 \phi^{q-1}  \\ & \qquad \qquad = 
2\left[ 8\pi \rho(\phi) \phi^{q-1} + \Lambda\phi^{q-1}\right]\, ;\\
& \frac{1}{2} \ck^* \left(\frac{1}{2N}\ck W\right) -\kappa \div_\phi^* \left( \frac{1}{N} \div_\phi (\phi^{2-q}V+Q) \right) = 8\pi j
\end{aligned}
\end{equation}
where the CKF compatibility condition \eqref{eq:ckfcompat-driftV-vac} becomes
\begin{equation}\label{eq:ckfcompat-driftV}
\kappa\int \frac{1}{N} \div_\phi (\phi^{2-q}V+Q) \div_\phi P\;  dV_g  = -8\pi \int_M j_a P^a dV_g
\end{equation}
for all conformal Killing fields $P$.

Analogously, CED-C equations with matter, generalizing system \eqref{eq:driftC-vac}, are
\begin{equation}\label{eq:driftC}
\begin{aligned}
-a\Lap \phi + R \phi & + \left|\sigma + \frac{1}{2N}\ck (\phi^{2-q}W + \phi^-q E) \right|^2 \phi^{-q-1} 
+ \kappa \left(\tau_* + \frac{1}{N\phi^{q}}\div_\phi(V + Q)) \right)^2 \phi^{q-1}  \\ & \qquad \qquad = 
2\left[ 8\pi \rho(\phi) \phi^{q-1} + \Lambda\phi^{q-1}\right] \, ;\\
& \frac{1}{2} \ck^* \left(\frac{1}{2N}\ck (\phi^{2-q}W + \phi^{-q}E)\right) -\kappa \div_\phi^* \left( \frac{1}{N} \div_\phi (V) \right) = 8\pi j,
\end{aligned}
\end{equation}
with compatibility condition
\begin{equation}
\int \frac{1}{4N}\ck (\phi^{2-q}W + \phi^{-q}E)\ck(\phi^{-q}F)\;dV_g = 8\pi \int j_a \phi^{-q} F^a dV_g
\end{equation}
for all divergence-free vector fields $F$.

\section{Near-CMC Solutions on Compact Manifolds using Drifts}\label{sec:small-drift-cpct}
In this section we prove the main result that, loosely stated, the drift method provides a good parameterization of solutions of the 
constraint equations on compact manifolds near CMC solutions, even when the metric admits conformal Killing fields. 

To begin, we characterize the CMC solutions with respect to drift parameters.
\begin{lemma}\label{lem:driftcmc} Suppose $(\ol g, \ol K)$ is a solution of the 
constraints equations \eqref{eq:constraints} with matter fields $(\ol \rho,\ol j,\Lambda)$, 
and let $\alpha$ be an arbitrary volume gauge. The solution is CMC if and only if 
\begin{enumerate}
\item the solution has zero volumetric drift measured by $\alpha$, and
\item for all conformal Killing fields $P$,
\begin{equation}\label{eq:Jisperp}
\int \ol j_a P^a  dV_{\ol g} = 0.
\end{equation}
\end{enumerate}
\end{lemma}
\begin{proof}
Write $\ol K= \ol A + \frac{\tau}{n} \ol g$, where $\ol A$ is trace-free with respect to $\ol g$, and then, applying Lemmas \ref{lem:yorksplit} 
and \ref{lem:yorksplitvol} with $N=dV_g/\alpha$, decompose further as
\begin{equation}\label{eq:Atau}
\ol A = \ol \sigma + \frac{1}{2N} (\ck \ol W),\ \ \tau = \tau_* + \frac{1}{N} \div \ol V,
\end{equation}
where $\ol \sigma$ is transverse-traceless, $\tau_*$ is constant, and $\ol W$ and $\ol V$ are vector fields.  With respect
to this decomposition, the momentum constraint reads
\begin{equation}\label{eq:mom-drift-V}
\frac{1}{2}\ck ^* \frac{1}{2N} \ck \ol W + \kappa \extd \frac{1}{N} \div \ol V = \ol j.
\end{equation}

First suppose that the solution is CMC. The expression for $\tau$ in \eqref{eq:Atau} implies that $\int N( \tau - \tau_*)\, dV_{\ol g} 
= \int \div \ol V\, dV_{\ol g} = 0$, and since $N>0$ everywhere, we see first that $\tau = \tau_*$ and then that $\div \ol V = 0$. 
Recall now from Definition \ref{def:driftmeas} that the volumetric drift measured by $\alpha$ 
is $\ol V + (\Ker \ck + \Ker \div)$.  Since $\ol V\in \Ker\div$, the solution has zero volumetric drift.  Moreover, multiplying equation
\eqref{eq:mom-drift-V} by a conformal Killing field and integrating by parts on the left-hand side yields \eqref{eq:Jisperp}.

Conversely, suppose the solution has zero volumetric drift measured by $\alpha$
and that \eqref{eq:Jisperp} holds. Since the solution has zero volumetric drift we can write $\ol V=E+Q$ where
$E$ is divergence-free and $Q$ is a conformal Killing field.  Observe that $\div \ol V = \div Q$.  Now multiply 
the momentum constraint \eqref{eq:mom-drift-V} by $Q$ and integrate by parts to get
\begin{equation}\label{eq:mom+Q}
0 = \int \ol j_a Q^a dV_g = -\kappa \int \frac 1 N (\div \ol V)(\div Q) \; dV_{\ol g}= -\kappa \int \frac 1 N (\div Q)^2 \; dV_{\ol g},
\end{equation}
i.e., $\div Q = 0$.  Finally, 
\begin{equation}
\tau = \tau_* + \frac{1}{N}\div \ol V = \tau_* + \frac{1}{N}\div Q = \tau_*,
\end{equation}
so the solution is CMC.
\end{proof}

Lemma \ref{lem:driftcmc} suggests that the volumetric form of the drift method is nicely adapted to generate CMC solutions.
Indeed, if a volumetric drift conformal data $(g,\sigma,\tau_*,V;\; N)$ generates a solution metric $\ol g = \phi^{q-2} g$,  then
the corresponding volumetric drift is $\drift[\ol g]{\phi^{2-q}V}$. This means that $V=0$ suffices; furthermore, at least when the
metric admits no conformal Killing fields, Lemma \ref{lem:driftrep} implies that $V=0$ is the only choice in $\mathcal D_g$
which results in a zero volumetric drift.    The only potential difficulty is that \eqref{eq:Jisperp} involves the unknown solution 
metric and the solution momentum density $\ol j$. Fortunately, the scaling law \eqref{eq:momentumtrans} for momentum density 
ensures that \eqref{eq:Jisperp} is conformally invariant, so this is not a real problem.

\begin{corollary}\label{cor:driftCMC}
Consider a CED-V data set $(g,\sigma,\tau_*,V;\; N)$  with $V\equiv 0	$ and a conformal matter distribution 
$(\rho(\cdot),j,\Lambda)$ that generates a solution $(\phi, W,Q)$ of the CED-V equations. The associated 
solution of the constraint equations is CMC if and only if 
\begin{equation}\label{eq:jcompat}
\int j_a P^a\;dV_g = 0
\end{equation}
for all conformal Killing fields $P$.  Moreover, when the solution is CMC, then $Q$ is  a true Killing field for the solution metric, 
$(\phi,W,\tilde Q\equiv 0)$ is also a solution of the CED-V equations that generates the same solution of the constraints, and 
$(\phi,W)$ solves the standard CTS-H equations \eqref{eq:CTS-H} with constant mean curvature $\tau = \tau_*$.
\end{corollary}
\begin{proof}
Let $(\phi,W,Q)$ be the solution of \eqref{eq:driftV2} corresponding to a solution $(\ol g, \ol K)$ of the
constraints. Since $V=0$ the volumetric drift of the solution is $\drift[\ol g]{\phi^{2-q}V}=0$, and Lemma \ref{lem:driftcmc} 
implies that the solution is CMC if and only if 
\begin{equation}\label{eq:jcompatinproof}
\int \ol j_a P^a dV_{\ol g} = 0
\end{equation}
for all conformal Killing fields $P$; here $\ol j$ is the physical momentum density given by 
\begin{equation}
\ol j = \phi^{-q} j.
\end{equation}
The physical volume form is $dV_{\ol g} = \phi^q dV_g$, and hence $\ol j dV_{\ol g} = j dV_g$.  So
\eqref{eq:jcompatinproof} holds for all conformal Killing fields $P$ if and only if the same is true for \eqref{eq:jcompat}.

Supposing now that the solution is CMC, equation \eqref{eq:jcompat} along with the choice $V\equiv 0$ implies that 
the CKF compatibility condition \eqref{eq:ckfcompat-driftV} reduces to
\begin{equation}
\int \frac{1}{N}(\div_\phi Q)(\div_\phi P) \; dV_g = 0
\end{equation}
for all conformal Killing fields $P$.  In particular, $\int (\div_\phi Q)^2N^{-1}\, dV_g=0$. But $\div_\phi = \div_{\ol g}$,
so $Q$ is a divergence-free conformal Killing field for $\ol g$, i.e. it is a true Killing field for $\ol g$.
Since $Q$ appears in the CED-V equations only via $\div_\phi Q$, we may as well take it to be zero and arrive at 
the same solution of the constraints. Finally, since $V\equiv 0$ as well, a quick inspection verifies that
the CED-V equations \eqref{eq:driftV2}  reduce to the CTS-H equations \eqref{eq:CTS-H} with $\tau\equiv \tau_*$.
\end{proof}

Corollary \ref{cor:driftCMC} shows that the CMC theory for the CTS-H equations
transfers directly to the CED-V equations.  In vacuum we have the tidy
and complete classification of CMC solutions completed in \cite{Isenberg:1995bi}.
Conversely, for a positive cosmological constant $\Lambda$
with $2\Lambda>\kappa\tau_*^2$ we have all the complexity demonstrated
in, e.g, \cite{Chrusciel:2015tk}.

\subsection{Near-CMC Solutions Parametrized by Small Volumetric Drift}

Given the natural connection between CMC solutions and volumetric drift, we first examine the construction 
of near-CMC solutions by perturbing to small volumetric drift.  We use the implicit function theorem in a fashion parallel
to that of \cite{Gicquaud:2014bu}, but with some technical features to handle conformal Killing fields.  Indeed,
in the presence of conformal Killing fields, the vector Laplacian $\frac{1}{2}\ck^*\frac{1}{2N}\ck: W^{k,p}(M,TM) \to W^{k-2,p}(M,T^*M)$
is Fredholm and has kernel equal to $\calQ$ and cokernel 
\[
\mathcal Q^\perp := \left\{\eta\in W^{k-2,p}(M,T^*M): \int_M \eta_a Q^a\;dV_g=0\text{ for all conformal Killing fields $Q$}\right\}.
\]
By modifying the domain and range slightly we can make this into an isomorphism
\begin{equation}\label{lem:divck}
\ck^* \frac{1}{2N} \ck :  W^{k,p}(M,TM)/\calQ \ra W^{k-2,p}(M,T^*M)\cap \calQ^\perp.
\end{equation}
Indeed, writing $[W]_{\calQ}$ for the projection of a $W^{k,p}$ vector field $W$ to the quotient space $W^{k,p}(M,TM)/\calQ$,
it is clear that $\ck[W]_\calQ =\ck W$ is well-defined.  
We also set
\begin{equation}
\calP: W^{k-2,p}(M,T^*M)\ra W^{k-2,p}(M,T^*M)\cap \calQ^\perp,
\label{lem:projectQ}
\end{equation}
to be the projection with kernel consisting of the conformal Killing covector fields; this is defined since 
$W^{k-2,p}(M,T^*M)\cap \calQ^\perp$ has finite codimension in $W^{k-2,p}(M,T^*M)$.

For the remainder of this section, fix a metric $g$, a lapse $N$, and a conformal matter distribution $(\rho(\cdot), j, \Lambda)$.
Before setting up an implicit function theorem argument, we define the following three functionals: 
\begin{itemize}
\item the Hamiltonian constraint
\begin{equation}\label{eq:CH}
\begin{aligned}
\MoveEqLeft C_H(\sigma, \tau_*, V;\;\phi,[W]_{\calQ},Q) = &\\
&-a\Lap \phi + R \phi - \left|\sigma + \frac{1}{2N}\ck [W]_{\calQ} \right|^2 \phi^{-q-1} 
+ \kappa \left(\tau_* + \frac{1}{N\phi^q}\div_\phi(\phi^{q-2}V + Q) \right)^2 \phi^{q-1} \\ & -
2(8\pi\rho(\phi) +\Lambda) \phi^{q-1};
\end{aligned}
\end{equation}
\item
the momentum constraint
\begin{equation}\label{eq:CM}
C_M(\sigma, \tau_*, V;\;\phi,[W]_{\calQ},Q) =
\frac{1}{2}\ck^*\frac{1}{2N}\ck [W]_{\calQ} - \calP\left[ \kappa \div_\phi^* \left( \frac{1}{N}\div_\phi(\phi^{q-2}V + Q)\right) +8\pi j\right]
\end{equation}
where $\calP$ is the projection \eqref{lem:projectQ};
\item the CKF compatability constraint
\begin{equation}\label{eq:CC}
C_C(\sigma, \tau_*, V;\;\phi,[W]_{\calQ},Q) =
P\mapsto \int\kappa\frac{1}{N} \div_\phi(\phi^{q-2}V+Q)\div_\phi(P) + 8\pi j_aP^a\;dV_g,
\end{equation}
where $P$ is an arbitrary conformal Killing field.  
\end{itemize}
There is a semicolon appearing in the arguments of these maps to separate those variables that are prescribed, 
the following spaces:
\begin{equation}\label{eq:Fdomain1}
(\sigma,\tau_*,V) \in [\ker \ck^*\subseteq W^{k-1,p}(M,S_2 M)] \times \Reals \times [\mathcal D_g\subseteq W^{k,p}(M,TM)],
\end{equation}
versus those that must be solved for,
\begin{equation}\label{eq:Fdomain2}
(\phi,[W]_\calQ,Q) \in W^{k,p}_+(M)\times (W^{k,p}(M,TM)/\calQ)\times 
\mathcal Q.
\end{equation}
The maps $C_H$, $C_M$, and $C_C$ take their values in $W^{k-2,p}(M)$, $W^{k-2,p}(M,T^*M)\cap\calQ^\perp$ and $\calQ^*$, respectively.

\begin{lemma}\label{lem:IFT-CEDV}
A triple $(\phi, W,Q)$  solves the CED-V equations
\eqref{eq:driftV2} for 
CED-V data $(g,\sigma,\tau_*,V;\;N)$
and conformal matter distribution $(\rho(\cdot),j,\Lambda)$ if and only if
\begin{equation}\label{eq:allzero}
\begin{aligned}
C_H(\phi,[W]_{\calQ},Q;\; \sigma, \tau_*, V) &= 0 \\
C_M(\phi,[W]_{\calQ},Q;\; \sigma, \tau_*, V) &= 0 \\
C_C(\phi,[W]_{\calQ},Q;\; \sigma, \tau_*, V) &= 0.
\end{aligned}
\end{equation}
\end{lemma}
\begin{proof}
If the definition of $C_M$ were not to involve the projection $\calP$ there would be nothing to do other than to observe 
that the distinction between $W$ and $[W]_{\calQ}$ is immaterial since $W$ only appears as an argument to $\ck$.
Hence it suffices to show that if $C_C=0$ then $C_M=0$ is equivalent to
\begin{equation}\label{eq:CM2}
\frac{1}{2}\ck^*\frac{1}{2N}\ck [W] - \div_\phi^* \left( \frac{1}{N}\div_\phi(V + Q)\right) -8\pi j = 0.
\end{equation}

Indeed, if $C_C=0$, then integration by parts shows that
\begin{equation}
-\div_\phi^* \left( \frac1{N} \div_\phi V\right) - 8\pi j \in \calQ^\perp.
\end{equation}
Hence
\begin{equation}
\calP\left[ \div_\phi^* \left( \frac{1}{N}\div_\phi(V + Q)\right) + 8\pi j\right] = 
\div_\phi^* \left( \frac{1}{N}\div_\phi(V + Q)\right) + 8\pi j,
\end{equation}
which is \eqref{eq:CM2}.
\end{proof}

\begin{theorem}\label{thm:smalldrift-conf}  
Consider volumetric drift parameters $(g,\hat \sigma,\hat \tau_*, \hat V;\; N)$ and and a conformal matter distribution $(\rho(\cdot),j,\Lambda)$ where $g$, $N$, and $\hat V$ have $W^{k,p}$ regularity, $\hat \sigma$ is of class $W^{k-1,p}$, $j$ is of
class $W^{k-2,p}$, and where $\rho$ satisfies the
energy scaling condition of Definition \ref{def:rhoscale}.

Suppose that $\hat V\equiv 0$ leads to a CMC
solution of equations \eqref{eq:driftV2} and additionally that
\begin{enumerate}
\item[i)] The CMC solution metric does not admit any true Killing fields.
\item[ii)] $\kappa \tau_*^2 \ge 2\Lambda$.
\item[iii)] Either $\kappa \tau_*^2 > 2\Lambda$, or $\sigma\not\equiv 0$, or the solution is not vacuum.\label{hyp:noscaling}
\end{enumerate}
Then there exists $\epsilon>0$ such that all conformal data $(g,\sigma,\tau_*,V;\; N)$ satisfying
\begin{equation}
||\sigma-\hat\sigma||_{W^{k-1,p}}+|\tau_*-\hat\tau_*| + ||V||_{W^{k,p}} < \epsilon
\end{equation}
generate a solution of system \eqref{eq:driftV2}, and the map from $(\sigma,\tau_*,V)$ to the associated 
solution of the constraint equations is smooth and injective.
\end{theorem}
\begin{remark}
As we show below, in the absence of matter fields, hypothesis i) is satisfied generically in the space of CMC solutions. 
\end{remark}
\begin{proof}
By conformal covariance of the CED-V equations, assume that the background metric $g$ is the CMC solution metric,
which means that the solution of the CED-V equations is $(\hat\phi,\hat W,\hat Q)$ with $\hat\phi\equiv 1$.  
Moreover Corollary \ref{cor:driftCMC} implies $\hat Q$ must be a true Killing field, hence $\hat Q\equiv 0$.  This 
simplifies various expressions later in the proof.  Although $\hat W$ does not have a simple expression, the momentum constraint implies
\begin{equation}\label{eq:jzerotoWzero}
-\ck^*\frac{1}{2N} \ck \hat W = j,
\end{equation}
which we also use in the sequel.

Define $\mathcal F= (C_H,  C_M,  C_C)$. 
By Lemma \ref{lem:driftrep} there is a neighborhood $\Phi$ of $1$ in $W^{k,p}_+(M)$ such that 
$V\mapsto \drift[\phi^{q-2}\hat g]{\phi^{2-q}V}$, from $\mathcal D_{\hat g}$ to $\Drift_{\phi^{q-2}\hat g}$,
is an isomorphism for any $\phi\in \Phi$.  We restrict the domain of $\mathcal F$ to these conformal factors;
it remains an open set in the Banach space \eqref{eq:Fdomain1}, \eqref{eq:Fdomain2}.

The map $\mathcal F$ is continuously differentiable and its derivative with respect to 
$(\phi,[W]_{\calQ},Q)$ at 
\begin{equation}
(\hat \sigma,\hat \tau_*, \hat V; \hat \phi, [\hat W]_\calQ, \hat Q) = 
(\hat \sigma,\hat \tau_*, 0; 1, [\hat W]_\calQ, 0)
\end{equation}
can be written as 
\begin{equation}\label{eq:linearizecpct}
DF(\delta\phi, \delta [W]_{\calQ}, \delta Q) = 
\begin{pmatrix}
-a\Lap + A & 
- 2\ip< \sigma+ \frac{1}{2N}\ck \hat W, \frac{1}{2N}\ck (\cdot) > & 2\kappa \tau_*\frac{1}{N} \div ( \cdot )  \\
0 & \frac{1}{2}\ck ^* \left(\frac{1}{2N} \ck (\cdot )\right) & \kappa\calP( \div^* \left( \frac{1}{N} \div (\cdot ) \right))  \\
0 & 0 & P \mapsto \kappa\int \frac{1}{N} \div( \cdot ) \div(P)\; dV 
\end{pmatrix}
\begin{pmatrix} \delta\phi \\ \delta [W]_\calQ \\ \delta Q \end{pmatrix}
\end{equation}
where 
\begin{equation}\label{eq:A}
A =  (q+2)\left|\sigma+\frac{1}{2N}\ck \hat W\right|^2 +(q-2)[\kappa\tau_*^2-2\Lambda]-16\pi[\rho'(1)+ (q-2)\rho(1)].
\end{equation}
Note that we have used the Hamiltonian constraint 
\begin{equation}\label{}
R - \left|\sigma+\frac{1}{2N}\ck \hat W\right|^2 + \kappa \tau_*^2 = 16\pi\rho(1)+2\Lambda
\end{equation}
to replace the scalar curvature that would otherwise have appeared in the expression for $A$.
From the block upper-triangular form of the matrix we conclude that $DF$ is invertible if 
each diagonal block is, and we treat each in turn. 

The operator 
\begin{equation}
-a\Lap + A  : W^{k,p}(M) \ra W^{k-2,p}(M),
\end{equation}
is invertible if $A\ge 0$, $A\not \equiv 0$.  Looking at the expression \eqref{eq:A} we have three terms to consider. First,
\begin{equation}
(q-2)[\kappa\tau_*^2-2\Lambda] \geq 0
\end{equation}
since $q>2$ (for any $n\ge 3$) and since $\kappa\tau^2_*\ge 2\Lambda$ by hypothesis. Next,
\begin{equation}
[\rho'(1)+ (q-2)\rho(1)]
\end{equation}
is the linearization of
\begin{equation}
\phi\mapsto \phi^{q-2}\rho(\phi)
\end{equation}
evaluated at $\phi\equiv1$.  By Definition \ref{def:rhoscale},  this is non-positive and hence $-16\pi[\rho'(1)+ (q-2)\rho(1)]\ge 0$.
The final summand of $A$ is 
\begin{equation}\label{eq:confKE}
\left| \sigma + \frac{1}{2N}\ck \hat W \right|^2,
\end{equation}
which is obviously nonnegative. Moreover, multiplying expression \eqref{eq:confKE} by $N$ and integrating yields
\begin{equation}
\int N \left| \sigma + \frac{1}{2N}\ck \hat W \right|^2
=\int N |\sigma|^2 + \frac{1}{4N}|\ck \hat W|^2,
\end{equation}
using that transverse-traceless tensors are $L^2$ orthogonal to the image of $\ck$.   Altogether, 
$A\equiv 0$ means that 
\begin{equation}
 \kappa\tau_*^2=2\Lambda,\ \ \sigma\equiv 0,\ \ \hat W\equiv 0,\ \mbox{and}\ \  \rho'(1)+(q-2)\rho(1)\equiv 0
\end{equation}
Equation \eqref{eq:jzerotoWzero} shows that $\hat W\equiv 0$ implies $j\equiv 0$; by Definition \ref{def:rhoscale}, 
if $\rho'(1)+(q-2)\rho(1)\equiv 0$ then $\rho(\cdot)\equiv 0$. From these we get that $\kappa\tau_*^2=2\Lambda$,
$\sigma\equiv 0$ and the solution is vacuum. Hypothesis \ref{hyp:noscaling} thus ensures that $A\not\equiv 0$.

The middle block of the matrix in \eqref{eq:linearizecpct} is invertible by the discussion around \eqref{lem:divck}.

Finally, for the last block, the symmetric bilinear form
\begin{equation}
B: \calQ \times \calQ \to \RR,\quad B(Q,P) = \int \frac{1}{N} \div Q \div P \; dV_g
\end{equation}
is nonnegative, and positive definite so long as $\calQ$ contains no true Killing fields, which are precisely
the divergence free elements in $\calQ$. This too holds under our assumptions. Therefore, the map 
\begin{equation}
Q\mapsto \int \frac{1}{N} \div Q \div (\cdot) \; dV_g 
\end{equation}
is an isomorphism from $\calQ$ to $\calQ^*$.

Taking Lemma \ref{lem:IFT-CEDV} into account, the implicit function theorem now provides the existence of the 
solution map for $(\sigma,\tau_*,V)$ sufficiently near $(\hat \sigma,\hat \tau_*,\hat V=0)$ 
in $W^{k,p}\times R\times \mathcal D_g$. It remains to establish the global injectivity.

Suppose $(\sigma,\tau_*,V)$ determines a solution $(\phi, W, Q)$ of the CED-V equations, and thereby
a solution $(\overline g,K)$ of the constraint equations. We demonstrate injectivity by showing that 
we can recover $(\sigma,\tau_*,V)$ from $(\overline g,K)$ under the hypothesis that $\phi\in \Phi$.

Setting $\overline N = \phi^{q} N$, apply Lemmas \ref{lem:yorksplit} and \ref{lem:yorksplitvol} to write
\[
K = \overline \sigma + \frac{1}{2\overline N} \ck_{\overline g} \overline W + \frac{\tau}{n}\overline g,\ \quad
\tau = \overline \tau_* + \frac{1}{\overline N}\div_{\overline g}(\overline V)
\]
where $\overline \sigma$ is transverse-traceless with respect to $\overline g$, $\overline \tau_*$ is constant, and $\overline V$
is a vector field.  On the other hand equations \eqref{eq:CEDV-reconstruct} and the conformal transformation laws for the 
divergence and conformal Killing operators imply
\[
K = \phi^{-2}\sigma + \frac{1}{2\overline N} \ck_{\overline g} 
W + \frac{\tau}{n}\overline g,\ \quad  
\tau = \tau_* + \frac{1}{\overline N}\div_{\overline g}(\phi^{2-q} V + Q).
\]
Since $\phi^{-2}\sigma$ is transverse-traceless with respect to $\overline g$, the uniqueness clauses of Lemmas
\ref{lem:yorksplit} and
\ref{lem:yorksplitvol} imply $\tau_*=\overline \tau_* $,
$\sigma=\phi^{2}\overline\sigma$ and that there
is a $\overline g$ divergence-free vector field $E$ such that
\[
\phi^{2-q} V + Q + E = \overline V.
\]
But this shows that we have agreement of drifts
\[
\drift[\ol g]{\phi^{2-q} V} = \drift[\ol g]{\ol V}.
\]
Since $\phi \in \Phi$,  the map $\mathcal D_{g} \rightarrow \Drift_{\overline g}$ given by $V\mapsto \drift[\ol g]{\phi^{2-q} V}$
is an isomorphism and $V\in\mathcal D_g$  is uniquely determined by $\overline V$.
\end{proof}



Proposition \ref{prop:CKFs-matter-sometimes} shows that given a CMC solution of vacuum constraint equations with a metric 
conformal to the round sphere, there exist inadmissible perturbations of the mean curvature.  By contrast,  Theorem \ref{thm:smalldrift-conf} 
shows that, so long as the CMC solution has no Killing fields, arbitrary small perturbations of drift and volumetric momentum 
produce nearby solutions. We now verify that this condition is generic among the CMC solutions within a conformal class.

\begin{proposition}\label{prop:genericity}
In the space of all CMC solutions to the vacuum constraint equations, the subset of pairs $(g,K)$ for which there 
are no Killing fields is open and dense. In fact, this is true even within a conformal class.
\end{proposition}
\begin{proof}
Let $(g,K)$ be any CMC solution and denote by $\calK_g$ and $\calQ_g$ the spaces of Killing and conformal Killing vector 
fields for $g$, respectively; thus
\[
\calQ_g = \left\{X: \ck X = 0\right\}, \quad \mbox{and} \qquad \calK_g = \left\{ X \in \calQ_g: \div_g X = 0\right\};
\]
of course $\calQ_{\tilde{g}} = \calQ_g$ for any metric $g' = \phi^{q-2}g$.

We first show that if $\calK_g = \{0\}$, then the same is true for any metric $g'$ near to $g$ in the $W^{k,p}$ topology. 
The second part is to prove that if $K_g$ is nontrivial, then there exist metrics $g'$ arbitrarily near $g$ in the
$W^{k,p}$ topology such that $\calK_{g'} = \{0\}$. 

To begin, observe that the $\calK_g$ is also characterized as the nullspace of the map 
\[
T_g: \calQ_g \longrightarrow \calQ_g, \ \ T_g \xi = \mathbb P \circ \div_g^* \circ \div_g,
\]
where $\mathbb P$ is the $L^2$ orthogonal projection from the space of symmetric two-tensors onto the finite
dimensional space $\calQ_g$; this follows easily from the identity $0 = \langle T_g \xi, \xi \rangle = 
|| \div_g \xi ||^2$ if $T_g \xi = 0$ and $\xi \in \calQ_g$. We henceforth
identify $\calQ_g$ with $\RR^N$ for some $N$. Observe also that $T_g$ depends in a real analytic way on $g$. 

For the first assertion, simply note that if $g$ admits no Killing fields, then $\ker T_g = \{0\}$, and this is an 
open condition in the space of all $W^{k,p}$ metrics, hence also in the space of metrics $g'$ which appear in a 
pair $(g',K')$ of CMC solutions of the constraint equations. 

As for the second assertion, suppose $\calK_{g_0}$ is nontrivial for some metric $g_0$ which appears in a CMC solution pair 
$(g_0, K_0 = \frac{\tau}{n} g_0 + \sigma_0)$. Without loss of generality we can assume that $\tau\neq 0$ and 
$\sigma_0\not\equiv 0$, for otherwise the CMC theory of the conformal method ensures we can
perturb to a nearby solution of the constraint equations satisfying this condition.
We consider families of solutions which arise by varying $\sigma$ in $\calU=W^{k,p}(M,S_{\rm tt})\setminus\{0\}$, but keeping
the conformal class fixed. From the CMC theory of the 
conformal method, since $\tau\neq 0$, for $\sigma\in\calU$ there is a well defined conformal factor $\phi(\sigma)$ 
obtained by solving the Lichnerowicz equation
\begin{equation}
-a \Delta_{0} \phi + R_0\phi - |\sigma|^2_{g_0} \phi^{-q-1} + \kappa \tau^2 \phi^{q-1}=0,
\end{equation}
and
\begin{equation}
(g_\sigma, K_\sigma)=\left(\phi^{q-2}g,\phi^{-2}\sigma+\frac{\tau}{n}\phi^{q-2} g\right)
\end{equation}
is a solution of the constraint equations. For simplicity, we write $T_\sigma$ instead of $T_{g_\sigma}$

Consider, for $j = 0, \ldots, N$, the subsets $\calF_j = \{\sigma \in \calU: \mathrm{rank}\, T_\sigma \leq j\}$. 
We claim that since $\phi$, and hence $g$, depends real analytically on $\sigma$, each $\calF_j$ is an analytic subvariety 
of finite codimension in $\calU$. Indeed, $\sigma$ lies in $\calF_j$ if and only if the determinant of every 
$(j+1)$-by-$(j+1)$ minor of $T_\sigma$ vanishes, and this is a finite number of polynomial conditions. By analyticity 
again, if the set $\calF_j^o := \calF_j \setminus \calF_{j-1}$ of TT tensors $\sigma$ where the rank of $T_\sigma$ is exactly 
$j$ has an interior point, then it is an open dense subset in $\calU$. Furthermore, 
$\calU$ is the union of the sets $\calF_j$, hence some $\calF_k^o$ must have interior, and hence is open and dense. 
The main conclusion follows if we can show that $k = N$, since $T_\sigma$ has full rank implies that its nullspace is trivial. 

Suppose that this is not the case, so $\calF_k^o$ is open and dense in $\calU$ for some $k < N$. We first show that 
there exists a submanifold in $\calU$ with finite codimension such that the nullspace of $T_\sigma$ is 
equal to the {\it same} $k$-dimensional subspace for every $\sigma$ in the submanifold. Indeed, consider the map $G: \calF_k^o 
\to G(k,N)$ into the Grassmanian of $k$-planes in $\RR^N$, which sends $\sigma$ to the nullspace of 
$T_\sigma$.  Let $\mathcal R$ be the image of $\calU$ under $G$. By construction, $\mathcal R$ is a subanalytic set in
$G(k,N)$, and hence itself admits a stratification, $\mathcal R = \sqcup \mathcal R_j$ where each $\mathcal R_j$
is a smooth $j$-dimensional submanifold. Suppose that $J$ is the maximal dimension of these strata, and let
$\calU' = G^{-1}(\mathcal R_J)$. This is an open dense set in $\calU$. 

The point of these maneuvers is to obtain a map $G' = G|_{\calU'}$ with maximal rank and image in a smooth manifold.  
We may now apply some familiar tools of differential topology. By the Sard-Smale theorem, there exists a full measure set 
of regular values of $G'$, and hence we may choose a $k$-plane $\Pi \subset \RR^N$ such that $\hat Z := (G')^{-1}(\Pi)$ 
is a smooth analytic submanifold of finite codimension in $\calU'$. In particular, the nullspace
of $\div_{g_\sigma}$ is the same $k$-dimensional subspace $\Pi \subset \calQ_g$ for all $\sigma \in \hat Z$. 

Fix $\hat \sigma_1 \in \hat Z$ and and write $\phi_1$ and $g_1$ for the corresponding conformal
factor and metric. Set $Z=\phi^{-2}\hat Z$, so $Z\subseteq W^{k,p}(M,\Stt(g_1))$ is a submanifold with finite codimension, 
and $\sigma_1=\phi_1^{-2}\sigma\in Z$.  The Lichnerowicz equation with $g_1$ as background metric is then
\begin{equation}
-a \Delta_1 \phi + R_1\phi - |\sigma|^2_{g_1} \phi^{-q-1} + \kappa \tau^2 \phi^{q-1}=0.
\label{Lichr2}
\end{equation}
By solving \eqref{Lichr2} for $\phi$, each $\sigma\in Z$ determines a  metric $g_\sigma=\phi^{q-2}g_1$ and second fundamental
form $K_\sigma$ solving the constraint equations. 
Moreover, let $H$ denote the connected component of the identity in the isometry group of $(M, g_1)$. This is a compact, connected Lie
group of positive dimension, and the quotient $M/H$ is an orbifold of strictly smaller dimension than $M$. 
Each $g_\sigma$ 
with $\sigma \in Z$ is invariant under $H$, or equivalently, the conformal factor $\phi(\sigma)$ (where
$g_\sigma = \phi^{q-2}g_1$) is invariant under $H$. This follows since $T_e H = \calK_{g_1}$ is actually constant
as $\sigma$ varies in $Z$. We show now that this leads to a contradiction.  

Suppose that $\sigma(\epsilon)$ is a one-parameter family of TT tensors lying in $Z$ with $\sigma(0)=\sigma_1$
and set $\eta = \dot\sigma(0)$. Differentiating the Lichnerowicz equation with respect to $\epsilon$ gives
\begin{equation}\label{Lichprime}
L \dot \phi = 2 \langle \sigma_1, \eta \rangle_{g_1}
\end{equation}
where 
\[
L := -a \Delta_1 + R_1 + (q+1) |\sigma_1|^2_{g_1} + (q-1) \kappa\tau^2 
\]
is the Frechet derivative of the Lichnerowicz equation at $\phi=1$.  Next differentiate \eqref{Lichprime} 
with respect to $X \in \calK_1$ to obtain
\[
L X\phi' = -[X,L]\phi +  2 X \langle \sigma_1, \eta \rangle_{g_1}
\]
Setting $\sigma=\sigma_1$ in equation \eqref{Lichr2}, the solution is $\phi=1$  and hence $R_1 + \kappa \tau^2 = |\sigma_1|^2_{g_1}$. 
The left side of this last relation is annihilated by any $X \in \calK_1$, hence so is the right, so it follows that all the
coefficient functions of $L$ are annihilated by $X$, and in particular $[X,L]=0$. Hence
\[
L X\phi' = 2 X \langle \sigma_1, \eta \rangle_{g_1}.
\]
On the other hand, $X(\phi(\sigma))=0$ for all $\sigma\in Z$
and therefore $X\phi'=0$.  Since $R_1 = |\sigma_1|^2_{g_1}
- \kappa \tau^2$ we can rewrite
\[
L=-a \Delta_1 + (q+2) |\sigma_1|^2_{g_1}  + (q-2) \kappa\tau^2
\]
to see that $L$ is invertible, and we conclude
that the pointwise inner product $\langle \sigma_1, \eta \rangle_{g_1}$ is constant along the $H$-orbits 
for every $\eta$ in the finite codimensional subspace $T_{\sigma_1}Z \subset W^{k,p}(M,\Stt(g_1))$. 

We now show that this last conclusion is absurd.  To this end, we use a construction presented in a neat and general
form in \cite{Delay:2012}, but in fact in fact in this finite regularity setting also following from  \cite{Mazzeo-edge}.
Namely, we claim that there exist $\eta \in W^{k,p}(M,\Stt(g_1))$ with arbitrarily small support. The basic principle is that the
operator $\div_g$ is left-elliptic, and under a certain hypothesis can be shown to be surjective acting between
symmetric trace-free two-tensors and vector fields (or $1$-forms) which vanish to some high order at the boundary
of some domain $\calO$. (This is proved in \cite{Delay:2012} using a weight function which vanishes exponentially in the
distance to $\partial \calO$, but follows from \cite{Mazzeo-edge} if one is content with weight functions which vanish
at any polynomial rate.)  We show how to apply this principle: suppose that $\chi \in \calC^\infty_0$ equals $1$ on 
an open set $\calO'$ which has closure contained in $\calO$ and which vanishes outside $\calO$. Denote by 
$\Omega$ the annular domain $\calO \setminus \overline{\calO'}$. If $\xi \in W^{k,p}(M,\Stt(g_1))$ is arbitrary, then 
$\div_{g_1}(\chi \xi) = \iota(\nabla \chi) \xi \in W^{k-1,p}$ has compact support in $\overline{\Omega}$. 
By \cite{Delay:2012, Mazzeo-edge}, there exists a symmetric trace-free $W^{k,p}$ two-tensor $\gamma$ supported in 
$\overline{\Omega}$ with $\div_{g_1} \gamma = \div_{g_1}(\chi \xi)$ if and only if $\iota(\nabla \chi) \xi$ 
is $L^2$ orthogonal to every $Y \in \calQ_g$, i.e., $\int_M \xi( \nabla \chi, Y) \, dV_{g_1} = 0$. To show
that this is satisfied here, observe that since $Y$ is conformal Killing and $\xi$ is trace-free,
\[
\div_{g_1} (\chi\; \iota(Y)\xi) = - \nabla^a_{g_1}(\chi \xi_{ab}Y^b) = 
- \xi(\nabla \chi, Y) + \chi \xi^{ab} \frac{1}{n}(\delta_{g_1} Y) (g_1)_{ab} = - \xi(\nabla \chi, Y)
\]
Integrating over $M$ yields the desired orthogonality. Hence $\chi \xi - \gamma \in W^{k,p}(M,\Stt(g_1))$ agrees with $\xi$
in $\calU$ and has support in $\calO$. 

Now choose disjoint open sets 
$\calO'_j$, $j = 1, \ldots, \ell$ such that
\begin{itemize}
	\item $\ell$ is larger than the codimension of $Z$,
	\item $\sigma_1\neq 0$ throughout each $\calO_j'$ (this is possible since $\sigma_1\not\equiv 0$),
	\item no integral curve of $X$ is contained in $\calO_j'$.
\end{itemize}
We can then apply the above construction to 
$\eta=\sigma_1$ on each $\calO'_j$ to obtain 
localizations $\sigma_{1j}$. 
Since $\ell$ is larger than the codimension of $Z$ 
there is a nontrivial linear combination
\[
\eta = \sum_j  b_j \sigma_{1j} = 0 \mod T_{\sigma_1}Z.
\]
That is, $\eta\in T_{\sigma_1}Z$.  Picking some $j$ such 
that $b_j\neq 0$, there is an integral
curve of $X$ which contains a point in $\calO_j'$ where
$\eta=\sigma_1\neq 0$.  
But this same integral curve is not contained in $\calO_j'$
and hence also contains a point on $\partial O_j'$ where
$\eta=0$.  It is then 
obvious that $\langle \sigma, \eta \rangle_{g_1}$ is not constant along the integral curve. 

This is the contradiction we desired. The proof is complete. 
\end{proof}

\subsection{Rescaling CED-V Conformal Parameters}
In \cite{Gicquaud:2014bu}, the authors observe that the far-from CMC solutions of the constraints 
constructed in \cite{HNT07b} and \cite{M09} can be considered as perturbations of solutions with $\tau\equiv 0$,
together with rescaling. In this section we examine how these arguments translate to the CED-V setting.

Starting from a pair $(g,K)$, consider a length $L>0$ and a rescaled pair $(\hat g,\hat K)=(L^2 g, L\, K)$.  If 
$(g,K)$ solves the constraints with physical matter distribution $(\rho,j,\Lambda)$, then $(\hat g,\hat K)$
solves the constraints with physical matter distribution 
\begin{equation}
(\hat \rho,\hat j,\hat \Lambda) = (L^{-2}\rho, L^{-1} j_{a}, L^{-2}\Lambda ).
\end{equation}
A straightforward computation establishes how this homothety scaling extends to CED-V parameters.
\begin{lemma}\label{lem:scaling}
Suppose $(\phi,W,Q)$ is a solution of the CED-V equations \eqref{eq:driftV2} for conformal
data $(g,\sigma,\tau_*,V;\;N)$ and conformal matter distribution $(\rho(\cdot ), j ,\Lambda)$.
For any $L>0$, 
\begin{equation}
(L^{\frac{n}{2}-1}\phi, L^{n-1} W, L^{n-1} Q)
\end{equation}
is a solution of system \eqref{eq:driftV2}
for conformal data 
\begin{equation}
(g,  L^{n-1}\sigma, L^{-1} \tau_*, L^{n-1} V;\; N)
\end{equation}
and conformal matter distribution 
\begin{equation}
(L^{-2}\rho
(L^{1-\frac{n}{2}}\;\cdot), L^{n-1} j, L^{-2}\Lambda).
\end{equation}
\end{lemma}

Lemma \ref{lem:scaling} should be compared with the analogous result for the CTS-H equations, where a solution $(\phi,W)$ for 
conformal data $(\sigma,\tau;\;N)$ scales to a solution $(L^{\frac n 2 -1} \phi, L^{n-1} W)$ for conformal data $(L^{n-1}\sigma,L^{-1}\tau;\;N)$.  
So for the CTS-H equations, we can effectively trade small $\tau$ for large $\sigma$ or vice-versa.  
Furthermore, if a solution with $\tau\equiv 0$ can be found, then nearby perturbatios and rescalings 
allow for arbitrary mean curvature. The situation is more complicated for the CED-V equations because there is an additional
parameter involved, but the principle is the same.  If we can find a solution with a parameter equal to zero, then we may hope 
to perturb off of it and rescale to obtain any value of the chosen parameter.  In the CMC case, volumetric drift is zero, and
hence we can obtain any desired volumetric drift.

\begin{corollary}\label{cor:anydrift}
Let $L>0$ be a constant and consider drift conformal data $(g,L^{n-1}\sigma,L^{-1}\tau_*, V;\;N)$
with conformal matter distribution $(L^{-2}\rho(L^{\frac{n}{2}-1}\;\cdot),L^{n-1}j,L^{-2}\Lambda)$, all with the regularity 
hypotheses considered in Theorem \ref{thm:smalldrift-conf}.
There exists a solution of the CED-V equations \eqref{eq:driftV2} for this data if $L$ is sufficiently large and if
all of the following hold:
\begin{itemize}
\item There exists a solution for the the CMC conformal data $(g,\sigma,\tau_*,0;\;N)$ with 
matter distribution $(\rho(\cdot),j,\Lambda)$.
\item There are no true Killing fields for the metric 
at the CMC solution.
\item $\kappa \tau_*^2 \ge \Lambda$
\item Either $\kappa \tau_*^2 > \Lambda$, or $\sigma\not\equiv 0$, or the matter distribution is not vacuum.
\end{itemize}
\end{corollary}
\begin{proof}
Consider the rescaled conformal data $(g, \sigma, \tau_*, L^{-1-n} V;\; N)$
with conformal matter distribution $(\rho(\cdot), j, \Lambda)$.  From the stated assumptions
we can apply Theorem \ref{thm:smalldrift-conf} to conclude that if $L$ is sufficiently large
(and hence $L^{-1-n}V$ is sufficiently small) 
there exists a solution $(\phi, W, Q)$  of system \eqref{eq:driftV2}
for this data. Let
\begin{equation}
(\hat \phi,\hat W,\hat Q)=(L^{\frac{n}{2}-1} \phi, L^{n-1}  W, L^{n-1} \hat Q)
\end{equation}
Lemma \ref{lem:scaling} implies $(\hat \phi,\hat W,\hat Q)$ is a solution of system
\eqref{eq:driftV2} for conformal data $(g,L^{n-1}\sigma,L^{-1}\tau_*,V;\;N)$
with matter distribution $(L^{-2}\rho(L^{1-\frac{n}{2}}\cdot), L^{n-1}j, L^{-2}\Lambda)$.
\end{proof}

In effect, Corollary \ref{cor:anydrift} provides a weak notion of the idea that we can obtain any volumetric
drift we please so long as we take the conformal momentum sufficiently large and the volumetric momentum sufficiently small. For maximal CMC solutions ($\tau_*=0$) an analogous
procedure shows that we can perturb to
an arbitrary volumetric momentum at the penalty of shrinking
both the conformal momentum and the volumetric drift.

\begin{corollary}
Under the same regularity hypotheses as Theorem 
\ref{thm:smalldrift-conf} suppose:
\begin{itemize}
\item There exists a solution for the maximal slice conformal data $(g,\sigma,0,0;\;N)$ with 
matter distribution $(\rho(\cdot),j,\Lambda)$.
\item There are no true Killing fields for the metric 
at the CMC solution.
\item $\kappa \tau_*^2 \ge \Lambda$
\item Either $\kappa \tau_*^2 > \Lambda$, or $\sigma\not\equiv 0$, or the matter distribution is not vacuum.
\end{itemize}
If $L>0$ is sufficiently small, then there exists a solution of the CED-V equations \eqref{eq:driftV2}
with prescribed conformal data $(g,L^{n-1}\sigma,\tau_*,L^{n+2}V;\;N)$ and matter distribution 
$(L^{-2}\rho(L^{\frac{n}{2}-1}\;\cdot),L^{n-1}j,L^{-2}\Lambda)$.
 \end{corollary}
\begin{proof}
Consider the rescaled conformal data $(g, \sigma, L \tau_*, L V;\; N)$
with matter distribution $(\rho(\cdot), j, \Lambda)$.  Since we have assumed 
that there exists a solution for the maximal slice data
$(g, \sigma, 0, 0;\; N)$, 
Theorem \ref{thm:smalldrift-conf} implies that if $L$ is sufficiently small
there exists a solution $(\phi, W, Q)$  of system \eqref{eq:driftV2}
for this data.  Rescaling as in the the proof of Corollary \ref{cor:anydrift},
we then find that that there exists a solution for 
conformal data $(g,L^{n-1}\sigma,L^{-1}\tau_*,L^{n+2}V;\;N)$
and matter distribution $(L^{-2}\rho(L^{1-\frac{n}{2}}\cdot), L^{n-1}j, L^{-2}\Lambda)$.
\end{proof}

\section{Extension to the AE and AH settings}
In this brief final section we indicate the modifications necessary to carry these results over to the two main 
noncompact settings common in this field, namely to sets of data which are asymptotically Euclidean (AE) or asymptotically 
hyperbolic (AH), respectively. 
(Extensions to other cases of interest, such as to compact manifolds with boundary, may be established by following the same overall approach.)

As is well known, in either of these cases, we may take advantage of known solvability results for the various linear
operators which appear in this paper, acting between appropriate weighted Sobolev spaces. Our intent here is not to be complete,
but rather just to briefly describe those parts of the arguments above that can be modified without further effort. In fact, 
there are no nontrivial conformal Killing fields vanishing at infinity in these settings, so the situation is somewhat simpler.  
On the other hand, this absence of conformal Killing fields implies both the standard conformal method and the drift method have 
perfectly adequate near-CMC theories for AE and AH initial data, and 
any potential advantages of the drift method is these cases would have
to arise for far-from CMC data. In the AH setting 
there are additional deeper questions concerning the `shear-free' condition
(see, e.g., \cite{AnderssonChrusciel1994}) but these have not been previously addressed even for the
standard conformal method and we leave their resolution for elsewhere.

\medskip

\noindent{\bf Asymptotically Euclidean Data:} 

We say that $(M, g, K)$ is an asymptotically Euclidean data set if there exists a compact region $K \subset M$
such that each of the finitely many components $E$ of $M \setminus K$ is diffeomorphic to $\RR^n \setminus B_R(0)$ 
for some $R > 0$, and using this diffeomorphism to give coordinates on each end, $g|_E = \delta + h$ where $\delta$ is the
Euclidean metric and $h_{ij} = \calO( |x|^{-1})$, along with corresponding estimates for the derivatives up 
to order $2 + \alpha$. At the same time, $K_{ij} = \calO( |x|^{-2})$ along with derivatives.  It is equally easy from
an analytic standpoint to include the somewhat more general case of asymptotically conic data. Here $M \setminus K$ 
is a finite union of ends $E$ where each $E$ is diffeomorphic to the `large end' of a Riemannian cone $C(Y)$,
with metric $dr^2 + r^2 k_Y$, where $(Y, k_Y)$ is a compact Riemannian manifold, and so that the corresponding 
estimates as above hold with this conic metric in place of the Euclidean metric. In either case, we also impose 
suitable decay conditions on matter fields.

The results that need to be modified in this new geometric setting are those which concern the global
solvability of certain elliptic problems.  The particular results that require different proofs are the York 
splitting Lemmas~\ref{lem:yorksplit} and \ref{lem:yorksplitvol}, and our main 
Theorem~\ref{thm:smalldrift-conf}.  In Theorem~\ref{thm:smalldrift-conf},
we decompose the conformal factor $\phi = 1 + u$, and 
because there are no conformal Killing fields vanishing at infinity
the map $F$ no longer involves the variable $Q$. Its linearization
from equation \eqref{eq:linearizecpct} becomes
\begin{equation}\label{eq:linearizecpctAE}
DF(\delta u, \delta W) = 
\begin{pmatrix}
-a\Lap + A & 
- 2\ip< \sigma+ \frac{1}{2N}\ck \hat W, \frac{1}{2N}\ck (\cdot) >   \\
0 & \frac{1}{2}\ck ^* \left(\frac{1}{2N} \ck (\cdot )\right) 
\end{pmatrix}
\begin{pmatrix} \delta u\\ \delta W \end{pmatrix}
\end{equation}
where, in vacuum, 
\begin{equation}\label{eq:A2}
A =  (q+2)\left|\sigma+\frac{1}{2N}\ck \hat W\right|^2 \ge 0.
\end{equation}

 For all of these adjustments we require the basic Fredholm properties of elliptic operators 
on asymptotically conic spaces, which appears, for example, in \cite{Mazzeo-edge} (and many other places). 
The main observation is that one needs to let such an operator act between spaces which are weighted by 
powers of $|x|$ at infinity. This theory is well-known, the elliptic operators involved in our application indeed invertible, and there are no unexpected issues.

\def\del{\partial}
\noindent{\bf Asymptotically Hyperbolic Data}

\nobreak
Another main setting in relativity is the asymptotically hyperbolic 
case; this generalizes the spacelike hyperboloid in Minkowski space, or equivalently, hyperbolic space. 
The natural generalization of this 
is the class of conformally compact asymptotically hyperbolic spaces. We say that $(M, g, K)$ is an asymptotically
hyperbolic data set if the following holds. First, $M$ is the interior of a smooth compact manifold with boundary
$\overline{M}$. The metric $g$ is of the form $\overline{g}/\rho^2$, where $\overline{g}$ is a metric smooth
and nondegenerate up to $\del \overline{M}$, and $\rho$ is a boundary defining function for the boundary which
satisfies $|\nabla^{\overline{g}} \rho|_{\overline{g}} = 1$ at $\rho = 0$. The tensor $K$ is again smooth up to
$\del \overline{M}$, and if we write $K = \sigma + (\tau/n) g$, then $\tau$ converges to a constant at $\rho = 0$.
It is straightforward to relax the regularity assumptions on the metric and second fundamental form. 

Here too there is a rich and well-developed analytic theory, again to be found in \cite{Mazzeo-edge} (parts of which again 
appear in many other places as well). We let the relevant operators act on function spaces which are weighted
by powers of $\rho$, or equivalently, by powers of $e^{-d}$, where $d$ is the Riemannian distance function on $M$,
e.g.\ distance to some fixed compact set in the interior.  We again observe that Laplace-type operators
are Fredholm when acting between weighted Sobolev spaces and that the three main results mentioned above
hold in this geometric setting as well.   The monograph \cite{Lee2006} 
works out the indicial roots for the relevant elliptic
operators in this setting; these indicial roots determine the precise ranges of weights on the function spaces.  



\section*{Acknowledgment}
This work was supported by NSF grant 1263544. 
The first author was also supported by NSF grants DMS-1262982 and DMS-1620366.
The third author was also supported by NSF grant DMS-1608223.

\bibliographystyle{amsalpha-abbrv}
\bibliography{smalldrift,smalldrift-DM}

\providecommand{\bysame}{\leavevmode\hbox to3em{\hrulefill}\thinspace}
\providecommand{\MR}{\relax\ifhmode\unskip\space\fi MR }
\providecommand{\MRhref}[2]{%
  \href{http://www.ams.org/mathscinet-getitem?mr=#1}{#2}
}
\providecommand{\href}[2]{#2}
\begin{thebibliography}{CBIM92}

\bibitem[AC94]{AnderssonChrusciel1994}
L.~Andersson and P.~T. Chru\'sciel, \emph{On ``hyperboloidal'' {C}auchy data
  for vacuum {E}instein equations and obstructions to smoothness of scri},
  Comm. Math. Phys. \textbf{161} (1994), no.~3, 533--568. \MR{1269390}

\bibitem[BE87]{Bourguignon:1987ge}
J.~P. Bourguignon and J.~P. Ezin, \emph{{Scalar Curvature Functions in a
  Conformal Class of Metrics and Conformal Transformations}}, Transactions of
  the American Mathematical Society \textbf{301} (1987), no.~2, 723--736.

\bibitem[CB09]{ChoquetBruhat:2009hv}
Y.~Choquet-Bruhat, \emph{{General relativity and the Einstein equations}},
  Oxford Mathematical Monographs, Oxford University Press, Oxford, 2009.

\bibitem[CBIM92]{ChoquetBruhat:1992fz}
Y.~Choquet-Bruhat, J.~Isenberg, and V.~Moncrief, \emph{{Solutions of
  constraints for Einstein equations}}, Comptes Rendus de l'Acad\'emie des
  Sciences. S\'erie I. Math\'ematique \textbf{315} (1992), no.~3, 349--355.

\bibitem[CBIY00]{ChoquetBruhat:2000iw}
Y.~Choquet-Bruhat, J.~Isenberg, and J.~W. York, \emph{{Einstein constraints on
  asymptotically Euclidean manifolds}}, Physical Review D \textbf{61} (2000),
  no.~8, 084034.

\bibitem[CG17]{Chrusciel:2015tk}
P.~T. Chru{\'{s}}ciel and R.~Gicquaud, \emph{{Bifurcating solutions of the
  Lichnerowicz equation}}, Annales Henri Poincar{\'e} \textbf{18} (2017),
  no.~2, 643--679.

\bibitem[De12]{Delay:2012}
E.~Delay, \emph{Smooth compactly supported solutions of some underdetermined
  elliptic pde, with gluing applications}, Communications in Partial
  Differential Equations \textbf{37} (2012), no.~10, 1689--1716.

\bibitem[DGH12]{Dahl:2012dk}
M.~Dahl, R.~Gicquaud, and E.~Humbert, \emph{{A limit equation associated to the
  solvability of the vacuum Einstein constraint equations by using the
  conformal method}}, Duke Mathematical Journal \textbf{161} (2012), no.~14,
  2669--2697.

\bibitem[GN14]{Gicquaud:2014bu}
R.~Gicquaud and Q.~A. Ngo, \emph{{On the far from constant mean curvature
  solutions to the Einstein constraint equations}}, arXiv.org (2014), no.~19,
  195014.

\bibitem[HNT08]{HNT07a}
M.~Holst, G.~Nagy, and G.~Tsogtgerel, \emph{Far-from-constant mean curvature
  solutions of {Einstein's} constraint equations with positive {Yamabe}
  metrics}, Physical Review Letters \textbf{100} (2008), no.~16,
  161101.1--161101.4, {{\sf arXiv:0802.1031 [gr-qc]}}.

\bibitem[HNT09]{HNT07b}
\bysame, \emph{Rough solutions of the {Einstein} constraints on closed
  manifolds without near-{CMC} conditions}, Communications in Mathematical
  Physics \textbf{288} (2009), no.~2, 547--613, {{\sf arXiv:0712.0798
  [gr-qc]}}.

\bibitem[HPP08]{Hebey:2008gk}
E.~Hebey, F.~Pacard, and D.~Pollack, \emph{{A variational analysis of
  Einstein-scalar field Lichnerowicz equations on compact Riemannian
  manifolds}}, Communications in Mathematical Physics \textbf{278} (2008),
  no.~1, 117--132.

\bibitem[IM96]{IsenbergMoncrief1996}
J.~Isenberg and V.~Moncrief, \emph{A set of nonconstant mean curvature
  solutions of the einstein constraint equations on closed manifolds},
  Classical and Quantum Gravity \textbf{13} (1996), no.~7, 1819.

\bibitem[IMP05]{Isenberg:2005we}
J.~Isenberg, D.~Maxwell, and D.~Pollack, \emph{{A gluing construction for
  non-vacuum solutions of the Einstein-constraint equations}}, Advances in
  Theoretical and Mathematical Physics \textbf{9} (2005), no.~1, 129--172.

\bibitem[IN77]{Isenberg:1977hy}
J.~A. Isenberg and J.~M. Nester, \emph{{Extension of the York field
  decomposition to general gravitationally coupled fields}}, Annals of Physics
  \textbf{108} (1977), no.~2, 368--386.

\bibitem[Is95]{Isenberg:1995bi}
J.~Isenberg, \emph{{Constant mean curvature solutions of the Einstein
  constraint equations on closed manifolds}}, Classical and Quantum Gravity
  \textbf{12} (1995), no.~9, 2249--2274.

\bibitem[KW74]{KW:1974}
J.~L. Kazden and F.~W. Warner, \emph{Curvature functions on compact
  2-manifolds}, Ann. of Math \textbf{99} (1974), 14--47.

\bibitem[Le06]{Lee2006}
J.~M. Lee, \emph{Fredholm operators and {E}instein metrics on conformally
  compact manifolds}, Mem. Amer. Math. Soc. \textbf{183} (2006), no.~864,
  vi+83. \MR{2252687}

\bibitem[Li44]{Lichnerowicz:1944}
A.~Lichnerowicz, \emph{{L'int\'egration des \'equations de la gravitation
  relativiste et le probl\`eme des $n$ corps}}, Journal de Math\'ematiques
  Pures et Appliqu\'ees. Neuvi\`eme S\'erie \textbf{23} (1944), 37--63.

\bibitem[Ma91]{Mazzeo-edge}
R.~Mazzeo, \emph{Elliptic theory of differential edge operators {I}},
  Communications in Partial Differential Equations \textbf{16} (1991), no.~10,
  1615--1664.

\bibitem[Ma09]{M09}
D.~Maxwell, \emph{A class of solutions of the vacuum {E}instein constraint
  equations with freely specified mean curvature}, Math. Res. Lett. \textbf{16}
  (2009), no.~4, 627--645.

\bibitem[Ma11]{Maxwell:2011if}
D.~Maxwell, \emph{{A model problem for conformal parameterizations of the
  Einstein constraint equations}}, Communications in Mathematical Physics
  \textbf{302} (2011), no.~3, 697--736.

\bibitem[Ma14a]{Maxwell:2014Drift}
\bysame, \emph{Initial data in general relativity described by expansion,
  conformal deformation, and drift}, arXiv:1407.1467, 2014.

\bibitem[Ma14b]{Maxwell:2014gx}
\bysame, \emph{{The conformal method and the conformal thin-sandwich method are
  the same}}, Classical and Quantum Gravity \textbf{31} (2014), no.~14, 145006.

\bibitem[Ma15]{Maxwell:2014Kasner}
\bysame, \emph{Conformal parameterizations of slices of flat {K}asner
  spacetimes}, Annales Henri Poincar\'e \textbf{16} (2015), no.~12, 2919--2954
  (English).

\bibitem[Ng15]{Nguyen:2015}
T.~C. Nguyen, \emph{Nonexistence and nonuniqueness results for solutions to the
  vacuum einstein conformal constraint equations}, arXiv:1507.01081, 2015.

\bibitem[PY03]{Pfeiffer:2003ka}
H.~P. Pfeiffer and J.~W. York, \emph{{Extrinsic curvature and the Einstein
  constraints}}, Physical Review. D. Third Series \textbf{67} (2003), no.~4,
  044022--044028.

\bibitem[PY05]{Pfeiffer:2005iz}
\bysame, \emph{{Uniqueness and Nonuniqueness in the Einstein Constraints}},
  Physical Review Letters \textbf{95} (2005), no.~9, 091101.

\bibitem[Yo73]{York:1973fl}
J.~W. York, \emph{{Conformally invariant orthogonal decomposition of symmetric
  tensors on Riemannian manifolds and the initial-value problem of general
  relativity}}, Journal of Mathematical Physics \textbf{14} (1973), no.~4,
  456--464.

\bibitem[Yo99]{YorkJr:1999jo}
\bysame, \emph{{Conformal {\textquotedblleft}thin-sandwich{\textquotedblright}
  data for the initial-value problem of general relativity}}, Physical Review
  Letters \textbf{82} (1999), no.~7, 1350--1353.

\end{thebibliography}
\end{document}